\documentclass[notitlepage,aps,letter,reprint,nofootinbib]{revtex4-1}

\usepackage{graphicx}
\usepackage{amsmath}
\usepackage{amssymb}
\usepackage{hyperref}
\usepackage{tikz}
\usepackage{overpic}
\usepackage{bbm}
\usepackage{float}
\usepackage{verbatim}
\usepackage{amsthm}
\usepackage{afterpage}
\usepackage{soul}
\usepackage{color}

\newcommand{\be}{\begin{eqnarray*}}
\newcommand{\ee}{\end{eqnarray*}}
\newcommand{\beq}{\begin{eqnarray}}
\newcommand{\eeq}{\end{eqnarray}}
\newcommand{\bk}{{\mathbf{k}}}

\newcommand{\br}{{\mathbf{r}}}

\newcommand{\ba}{{\mathbf{a}}}
\newcommand{\bb}{{\mathbf{b}}}

\newcommand{\bG}{{\mathbf{G}}}

\newcommand{\dd}{\mathrm{d}}

\newcommand{\bx}{\hat{{\mathbf{x}}}}
\newcommand{\by}{\hat{{\mathbf{y}}}}

\newcommand{\ket}[1]{\left|{#1}\right\rangle}
\newcommand{\bra}[1]{\left\langle{#1}\right|}

\newcommand{\id}{\mathbb{I}}

\newtheorem{prop}{Proposition}

\DeclareMathOperator{\sgn}{sgn}

\begin{document}
\title{Interacting Floquet topological phases in three dimensions}
\author{Dominic Reiss}
\author{Fenner Harper}
\author{Rahul Roy}
\affiliation{Department of Physics and Astronomy, University of California, Los Angeles, California USA}
\date{\today}
\begin{abstract}

{In two dimensions, interacting Floquet topological phases may arise even in the absence of any protecting symmetry, exhibiting chiral edge transport that is robust to local perturbations. We explore a similar class of Floquet topological phases in three dimensions, with translational invariance but no other symmetry, which also exhibit anomalous transport at a boundary surface. By studying the space of local 2D unitary operators, we show that the boundary behavior of such phases falls into equivalence classes, each characterized by an infinite set of reciprocal lattice vectors. In turn, this provides a classification of the 3D bulk, which we argue is complete. We demonstrate that such phases may be generated by exactly-solvable `exchange drives' in the bulk. In the process, we show that the edge behavior of a general exchange drive in two or three dimensions can be deduced from the geometric properties of its action in the bulk, a form of bulk-boundary correspondence.}
\end{abstract}
\maketitle

\section{Introduction}

Driving a system periodically in time can generate remarkable behavior with an intrinsically dynamical character. In this rapidly evolving field of Floquet systems, recent advances include the prediction of phases which exhibit an analog of symmetry breaking in the time domain, known as discrete time crystals or $\pi$-spin glasses~\cite{Khemani:2016gd,Else:2017kg,vonKeyserlingk:2016ea,vonKeyserlingk:2016ev,Else:2016gf,Yao:2017bu}, as well as a host of novel topological phases that lie beyond any static characterization~\cite{Kitagawa:2010bu,Jiang:2011cw,Rudner:2013bg,Thakurathi:2013dt,Asboth:2014bg,Nathan:2015bi,Titum:2015fl,Carpentier:2015dn,Fruchart:2016hk,Titum:2016km,Roy:2017cs,vonKeyserlingk:2016bq,Else:2016ja,Potter:2016cr,Roy:2016ka,Potter:2017dp,Roy:2017ct}. These theoretical works have been complemented by significant experimental advances, with analogs of Floquet topological phases being realized in a variety of different settings~\cite{Kitagawa:2012gl,Rechtsman:2013fe,Jotzu:2015kz,JimenezGarcia:2015kd,Maczewsky:2016fj,Zhang:2017ci,Choi:2017ho}.

A particularly surprising set of Floquet topological phases are those which are robust even in the absence of symmetry~\cite{Rudner:2013bg,Po:2016iq,Harper:2017ce,Po:2017ex,Fidkowski:2017us}. In the presence of interactions, 2D systems in this class have been shown to exhibit robust Hilbert space translation at the boundary of an open system~\cite{Po:2016iq,Harper:2017ce}, and may be combined with bulk topological order to generate Floquet enriched topological phases~\cite{Po:2017ex,Fidkowski:2017us}. Despite their range of novel features, systems in this class have so far only been studied in 2D; in this paper, we set out to find and classify the Floquet topological phases that exist in 3D, under the assumption of translational invariance.

Similar to the classification of the related 2D phases, our approach is to first identify the distinct types of boundary behavior that these 3D Floquet systems may exhibit. By invoking ideas from Ref.~\cite{Gross:2012fm}, we find that local, translationally invariant unitary operators in two dimensions form distinct equivalence classes with representative `shift' (or translation) actions. In turn, this boundary classification partitions the space of 3D unitary evolutions in the bulk into distinct classes. Each class may be labeled by a topological invariant (in this case, an infinite set of reciprocal lattice vectors), with drives that are members of the same class being topologically equivalent at a boundary. We construct exactly solvable bulk drives which populate these equivalence classes, and in the process, identify a geometric property of such a drive that determines its anomalous behavior at an arbitrary boundary, a result that also applies to 2D. We argue that there are no intrinsically 3D Floquet topological phases (without symmetry), making this classification complete.

The structure of this paper is as follows: We begin with a brief review of Floquet systems and phases in Sec.~\ref{sec:Prelim} and provide some additional motivation for the work. In Sec.~\ref{sec:edgeclassification}, we describe and classify local, translationally invariant unitary operators with no symmetry in two dimensions, and show how this classification may be applied to the boundaries of 3D Floquet systems. In Sec.~\ref{sec:2D}, we describe a modification of the exactly solvable `exchange drives' introduced in Refs.~\onlinecite{Rudner:2013bg,Po:2016iq,Harper:2017ce}, and show that these have geometric properties directly related to their action at a boundary. In Sec.~\ref{sec:3D}, we extend these models to 3D, and demonstrate that they may be used to generate boundary behavior from all equivalence classes. We summarize and discuss our results in Sec.~\ref{sec:Discussion}.

\section{Preliminary Discussion \label{sec:Prelim}}

We begin by recalling some concepts from the study of time-dependent systems that we will use throughout the paper. We are primarily interested in Floquet systems, whose Hamiltonians are periodic in time (with $H(t+T)=H(t)$). The behavior of such a system is captured by the unitary time-evolution operator, defined by
\begin{equation}
U(t)=\mathcal{T}\exp\left[-i\int_0^tH(t')\dd t'\right],
\end{equation}
where $\mathcal{T}$ indicates time ordering. Although the system Hamiltonian can in general vary continuously with time, the models we consider in this paper will have Hamiltonians that are piecewise constant. For these systems, the complete unitary evolution operator is a product of unitary evolutions corresponding to each step, applied in chronological order.

We will classify these dynamical systems using the homotopy approach introduced in Ref.~\onlinecite{Roy:2017ct}, which is concerned with identifying topologically distinct paths $U(t)$ within the space of unitary evolutions. This framework has the advantage that it disentangles questions about the topology of the path $U(t)$ from questions about the stability of the resulting phase. For example, interacting Floquet systems are believed to be inherently unstable to heating, since energy is pumped into the system with every driving cycle~\cite{Lazarides:2014ie,DAlessio:2014fga,Ponte:2015hm}. To prevent heating to infinite temperature, strong disorder may be added so that the system becomes many-body localized~\cite{DAlessio:2013fv,Ponte:2015dc,Lazarides:2015jd,Abanin:2015bc,Abanin:2016eva,Zhang:2016hy} (see Ref.~\onlinecite{Nandkishore:2015kt} for a review of many-body localization (MBL)). In the homotopy approach, the topology of an evolution is well defined in the absence of MBL, even if MBL may be necessary in a physical realization of the model \cite{Roy:2017ct}.

The homotopy approach also allows a distinction to be made between static topological order, which depends only on the end point of the unitary evolution, and inherently dynamical topological order, which depends on the complete path of the evolution $U(t)$. This latter kind of order can be completely classified by studying a subset of unitary evolutions known as unitary loops~\cite{Roy:2017ct}, which, for a closed system, satisfy $U(0)=U(T)=\id$. For an open system, however, a unitary loop will not necessarily return to the identity, but may instead act nontrivially in a region near the boundary: We refer to the nontrivial action of a unitary loop restricted to this region as the `effective edge unitary'. Explicitly, we may write the closed system Hamiltonian as
\beq
H_{\rm closed}(t) &=& H_{\rm open}(t)+H_{\rm edge}(t),
\eeq
where $H_{\rm edge}$ connects sites across a boundary and $H_{\rm open}$ connects sites away from the boundary. We may then evolve with $H_{\rm open}(t)$ for a complete cycle to obtain the effective edge unitary $U_{\rm eff}$. Since $U_{\rm eff}$ acts as the identity in the bulk, we can restrict our attention to the boundary system on which the unitary acts non-trivially.

In this paper, our first aim is to obtain a complete characterization of effective edge unitaries described by local, two-dimensional unitary operators with translational invariance. We will then show that these distinct effective edge unitaries may be used to classify unitary loops in the 3D bulk, and will obtain an explicit set of loop drives which generate the different boundary behaviors. Although the unitary loops we introduce may seem somewhat fine-tuned, we will argue that any chiral Floquet phase is topologically equivalent to one of these representative drives, in the sense that their edge behaviors are equivalent.

%%%%%
%%%%%
\section{Edge classifications in 2D and 3D}\label{sec:edgeclassification}
Dynamical phases of 2D Floquet systems with no symmetry were classified based on their boundary behavior in Refs.~\onlinecite{Po:2016iq,Harper:2017ce}, building on a rigorous classification of 1D unitary operators from Ref.~\onlinecite{Gross:2012fm}. The aim of this paper is to obtain a similar classification of 3D Floquet systems by considering the distinct behaviors that may arise at a 2D boundary. To this end, we now briefly review the classification of unitary operators at a 1D boundary, before going on to discuss the 2D case.

\subsection{Effective unitary operators at a 1D boundary}\label{sec:edgeclassification2d}
As motivated in Sec.~\ref{sec:Prelim}, the edge action of a 2D Floquet system is fully described by a 1D unitary operator, $U_{\rm eff}$. Since the underlying Hamiltonian which generates the evolution should be physically motivated, the only restriction on the form of $U_{\rm eff}$ is that it should be \emph{local}---i.e., it should map local operators onto other local operators. There is no requirement, however, that it be possible to generate $U_{\rm eff}$ with a local \emph{one-dimensional} Hamiltonian. This potential anomaly partitions the space of 1D edge behaviors into different equivalence classes.

In Ref.~\onlinecite{Gross:2012fm}, 1D unitary operators of this form were classified according to the `net flow of quantum information' through the system. It was shown that this flow may be characterized by a discrete, locally computable index, which we refer to as the GNVW index. Unitaries within each resulting class are equivalent up to locally generated (in 1D) unitary transformations of finite depth. In the context of Floquet systems, these equivalence classes correspond to effective edge unitaries with distinct topological invariants.

We now review the construction and interpretation of the GNVW index of a 1D unitary, ${\rm ind}(U)$. First, we imagine cutting an (infinite) 1D system into left and right halves. We then choose a finite (but large) set of sites immediately to the left and to the right of the cut and denote these subsystems as $L$ and $R$, respectively. The GNVW index compares the extent to which the observable algebra in $L$ is mapped onto the observable algebra in $R$, and vice versa, by the action of the 1D unitary. 

Explicitly, we define the observable algebra on subsystems $L$, $R$ and their union $L+R$ to be $\mathcal{A}_L$, $\mathcal{A}_R$, and $\mathcal{A}$, respectively. The matrix units $\ket{e_{ij}}\cong\ket{i}\!\bra{j}$, with $\ket{i}$ and $\ket{j}$ states from the appropriate region, form a suitable basis for the observable algebra. A unitary operator $U$ acts on a member of an observable algebra through conjugation: i.e., the unitary action $\alpha_U$ on some element $M$ is $\alpha_U(M) = U M U^{-1}$. Finally, we define the normalized trace $\mathbf{Tr}$ for an operator algebra $\mathcal{A}$ with dimension $d$ as $\mathbf{Tr}(M) = \frac{1}{d}\textrm{Tr}(M)$, for any $M \in \mathcal{A}$, with $\textrm{Tr}$ the usual operator trace. 

With these definitions, the overlap of two subalgebras $\mathcal{B}_{1/2} \subset \mathcal{A}$ is given by 
\begin{equation}
\eta(\mathcal{B}_1,\mathcal{B}_2) = \sqrt{\mathbf{Tr}(P_1 P_2)},
\end{equation}
where the trace is taken over the algebra $\mathcal{A}$, and $P_{n}$ are projectors defined through $P_n=d_n\sum_{ij}\ket{e_{ij}}\!\bra{e_{ij}}$ (with $\ket{e_{ij}}\in\mathcal{B}_n$ and $d_n$ the dimension of $\mathcal{B}_n$). The value of $\eta(\mathcal{B}_1,\mathcal{B}_2)$ is always greater than or equal to one, with equality holding only when $\mathcal{B}_1$ and $\mathcal{B}_2$ commute~\cite{Gross:2012fm}.

In terms of $\eta$, the GNVW index of a unitary operator $U$ is given by the ratio
\begin{equation}\label{eq:GVNWind}
\textrm{ind}(U) = \frac{\eta(\alpha_U(\mathcal{A}_L),\mathcal{A}_R)}{\eta(\alpha_U(\mathcal{A}_R),\mathcal{A}_L)}.
\end{equation}
In Ref.~\onlinecite{Gross:2012fm} it is shown that ${\rm ind}(U)$ is always a positive rational number, $p/q$. In addition, the value of the index is independent of the choice of $L$ and $R$ (as long as they are sufficiently large) and independent of the location of the cut within the system. Importantly, $\textrm{ind}(U)$ is robust against unitary evolutions generated by local 1D Hamiltonians, and therefore defines a set of equivalence classes enumerated by positive rational numbers.

Each equivalence class has a representative unitary operator that may be defined in terms of `shifts'. A shift $\sigma$ is a unitary operator which uniformly translates the local Hilbert space on each site to the right by one site. Explicitly, if $\mathcal{H}_x$ is the Hilbert space on site $x$, then $\sigma\mathcal{H}_x=\mathcal{H}_{x+1}$. The representative unitary operator corresponding to the equivalence class with index $p/q$ is the tensor product $\sigma_p\otimes\sigma^{-1}_q$, which is a shift to the right of a Hilbert space with dimension $p$ combined with a shift to the left of a Hilbert space with dimension $q$. A generic (local) 1D unitary operator can always be brought to a representative shift of this form through the action of a finite-depth quantum circuit. In the context of 2D Floquet systems, these representative shift unitaries correspond to the chiral transport of a many-body state around the 1D boundary~\cite{Po:2016iq,Harper:2017ce}.

%%%
\subsection{Effective unitary operators at a 2D boundary}\label{sec:edgeclassification3d}

We now turn our attention to the edge action of 3D unitary loops, which may be described by effective unitary operators that are two dimensional. Without translational invariance, there is a large set of distinct 2D edge behaviors that could arise---for example, we could stack different shift drives $\sigma_p$ in parallel rows. In this paper we restrict the discussion to the manageable translationally invariant case, and leave a more general study to future work. 

The effective edge unitary ($U_{\rm eff}$) of a 3D unitary loop is a local unitary operator which acts on a quasi-2D boundary region. We assume that it is translationally invariant but that it may not be possible to generate $U_{\rm eff}$ using a local 2D Hamiltonian that acts only within the boundary region. Motivated by Refs.~\onlinecite{Po:2016iq,Harper:2017ce}, our approach will be to first classify local 2D unitary operators $U$ satisfying these properties, before using this boundary classification to infer a classification of the 3D bulk.

Without loss of generality, we choose $U$ to act on a Hilbert space which is a tensor product of $d$-dimensional Hilbert spaces located at each site of an (infinite) 2D Bravais lattice. Since $U$ is local, it has some Lieb-Robinson length $\lambda_{\rm LR}$ \cite{Lieb:1972wy}, and we assume for simplicity that the action of $U$ is strictly zero for separations greater than this length. 

In order to import some of the results from the 1D classification, we will treat the infinite 2D boundary as the limiting case of a sequence of quasi-1D cylindrical `periodic systems'. Given a lattice vector $\mathbf{r}$ and sufficiently large integer $N$ such that $|N\mathbf{r}| \gg \lambda_{\rm LR}$, we define a { periodic} system by identifying all lattice sites that are separated by an integer multiple of $N\mathbf{r}$. This periodic system can be thought of as having a compact dimension along the $\mathbf{r}$-direction with period $N\mathbf{r}$ and an extended dimension along any primitive lattice vector $\mathbf{r}'$ which is linearly independent to $\mathbf{r}$. We denote the unitary restricted to this periodic system as $U_{N,\mathbf{r}}$; since $U$ is translationally invariant and local, this restriction is always well defined.

We may now compute the GNVW index along the compact dimension of the periodic system. By defining a cut along $\mathbf{r}$, we divide this system in two halves ($L$ and $R$) as illustrated in Fig.~\ref{fig:cylinderLR}. The index, ${\rm ind}(U_{N,\mathbf{r}},\mathbf{r})$, associated with this cut can be calculated by viewing the system as a 1D edge (by grouping sites along $\mathbf{r}$) and using Eq.~\eqref{eq:GVNWind}.

\begin{figure}[t]
\centering
\includegraphics[width=\linewidth]{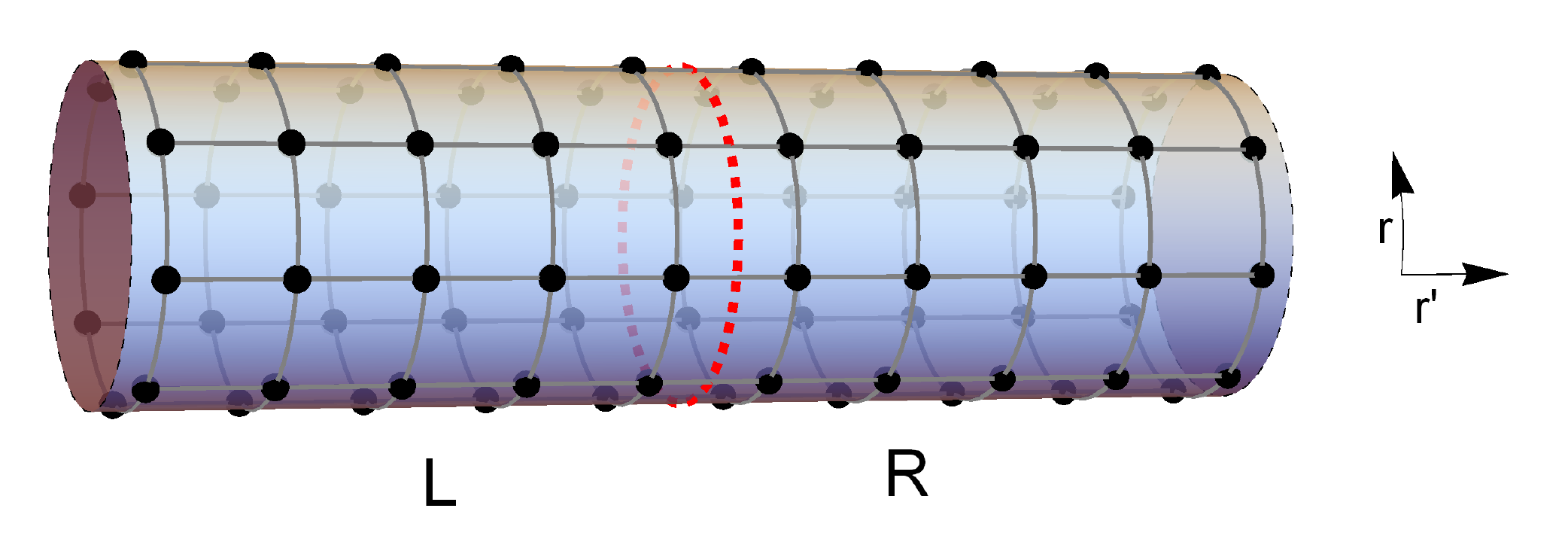}
\caption{Illustration of a choice of cut (red dashed line) which divides a 2D system with a compact dimension along $\mathbf{r}$ into the two halves $L$ and $R$. By grouping the Hilbert space of $N$ sites along the compact dimension, one may calculate the 1D GNVW index. This index is invariant under translations of the cut by $\mathbf{r'}$.} 
\label{fig:cylinderLR}
\end{figure}

The index ${\rm ind}(U_{N,\mathbf{r}},\mathbf{r})$ does not depend on the location of the cut, due to translational invariance in the $\mathbf{r}'$ direction. The value of ${\rm ind}(U_{N,\mathbf{r}},\mathbf{r})$ may, however, depend on the extent of the compact dimension $N\mathbf{r}$: If the compact dimension is made larger, then more information can flow across it. We therefore define a \emph{scaled} additive index 
\begin{equation}
\label{eq:defscaledadditiveindex}
\nu(\mathbf{r}) = \lim_{N \rightarrow \infty} \frac{1}{N} \log \textrm{ind}(U_{N,\mathbf{r}},\mathbf{r}),
\end{equation}
where the size of the periodic system is increased by taking the limit $N\to\infty$ for a fixed lattice vector $\mathbf{r}$. This limit defines a sequence of periodic systems which tends towards the infinite plane. { We expect the index $\textrm{ind}(U_{N,\mathbf{r}},\mathbf{r})$ to scale as a power of $N$ due to translation invariance, and we consequently expect $\nu(\mathbf{r})$ to be finite.}

 Since ${\rm ind}(U_{N,\mathbf{r}},\mathbf{r})$ is always a rational number \cite{Gross:2012fm}, the scaled additive index can be equivalently written as a sum over primes $p$, 
\begin{equation}
\label{eq:indexdecomp}
\nu(\mathbf{r}) = \sum_p n_p(\mathbf{r}) \log p,
\end{equation}
with integral coefficients $n_p$.

\begin{figure}[t]
\centering
\includegraphics[scale=0.32]{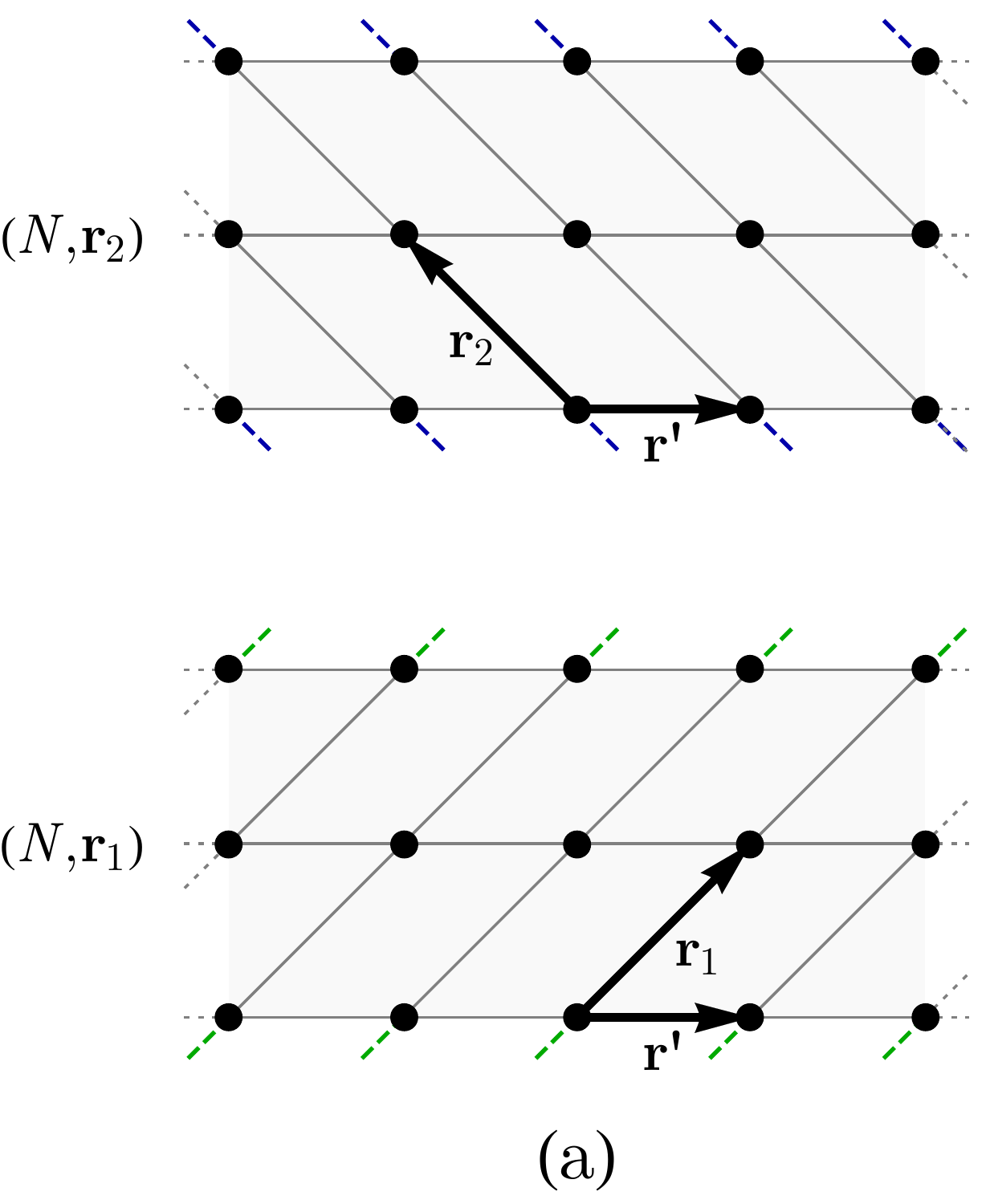}
\includegraphics[scale=0.32]{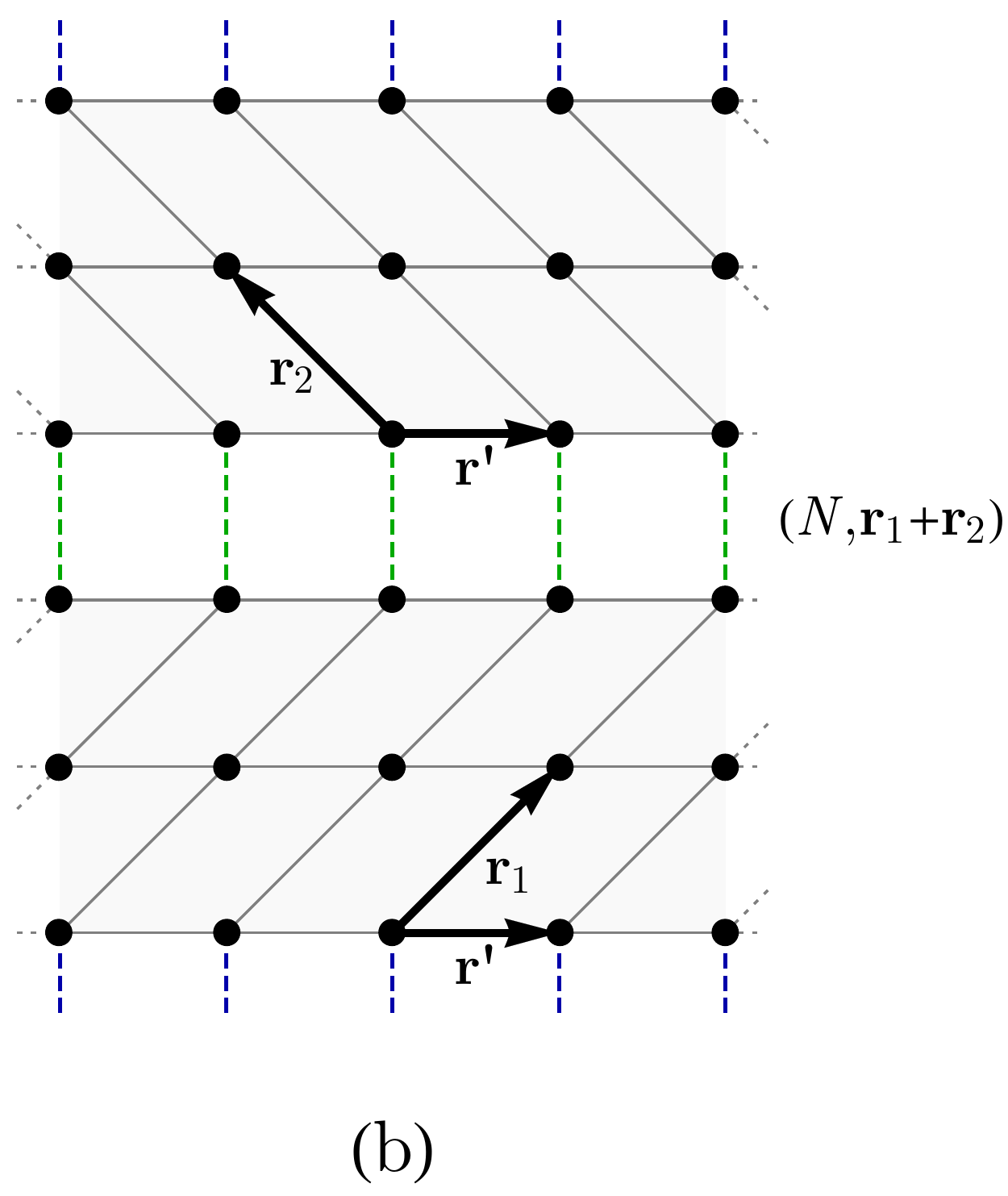}
\caption{The construction of a system from two periodic systems with common extended dimension $\mathbf{r}'$. (a) The lower periodic system has {compact} dimension along $\mathbf{r}_1$, with sites connected by the dashed green lines identified. Similarly, the upper system has periodic boundary conditions in the $\mathbf{r}_2$ dimension, with sites connected by dashed blue lines identified. (b) By cutting each system along a 1D sublattice in the $\mathbf{r}'$ dimension, and identifying sites between the two systems along the cuts (connected by dashed green/blue lines) we construct a system with periodic dimension along $\mathbf{r}_1+\mathbf{r}_2$.} 
\label{fig:join2dboundaries}
\end{figure}

We now investigate the relationship between different $\nu(\mathbf{r})$ with different choices of $\mathbf{r}$. We consider three periodic systems defined by $(N,\br_1)$, $(N,\br_2)$ and $(N,\br_1+\br_2)$, and write the action of the unitary restricted to each of these as $U_{N,\br_1}$, $U_{N,\br_2}$ and $U_{N,\br_1+\br_2}$, respectively. We now construct a fourth system, as shown in Fig.~\ref{fig:join2dboundaries}, by cutting the systems $(N,\br_1)$ and $(N,\br_2)$ each along a sublattice generated by some $\mathbf{r'}$ and reconnecting them along this line. The reconnection is carried out by restoring local terms such that the final system is compact along the $\mathbf{r}_1 + \mathbf{r}_2$ direction with length $N(\mathbf{r}_1 + \mathbf{r}_2)$. We write the action of the unitary on this composite system as $U'_{N,\br_1+\br_2}$, and note that further than $\lambda_{\rm LR}$ away from either cut, the action of $U'_{N,\br_1+\br_2}$ is identical that of $U_{N,\mathbf{r}_1+\mathbf{r}_2}$. 

We now argue that both of these unitaries correspond to the same index $\nu(\br_1+\br_2)$ and further, that $\nu(\br_1+\br_2)=\nu(\br_1)+\nu(\br_2)$. First, since $U'_{N,\br_1+\br_2}$ and $U_{N,\br_1+\br_2}$ differ (if at all) only in the vicinity of the two horizontal cuts used in defining the system, we must have
\begin{equation}
\label{eq:uuprimeindex}
{\rm ind}\left(U_{N,\br_1+\br_2}\right)=\delta\times{\rm ind}\left(U'_{N,\br_1+\br_2}\right),
\end{equation}
where $\delta$ is the contribution to the index caused by rejoining local terms in the unitary action across the cuts (to be discussed below).

We now construct a 1D cell structure for these systems compatible with both a `triangular' slice along the $\mathbf{r}_1$ direction followed by the $\mathbf{r}_2$ direction, and a `linear' slice along the $\mathbf{r}_1 + \mathbf{r}_2$ direction, as shown in Fig.~\ref{fig:cutcells}. The index ${\rm ind}(U)$ computed for a given unitary must be the same for either of these cuts, from the properties of the GNVW index \cite{Gross:2012fm}. Choosing the triangular slice, we see that the unitary $U'_{N,\br_1+\br_2}$ acts like either $U_{N,\br_1}$ or $U_{N,\br_2}$ away from the horizontal cuts. Overall, this means that
\begin{equation}
\label{eq:multiplicativeindexvectoradd}
\textrm{ind}(U_{N, \mathbf{r}_1 + \mathbf{r}_2})= \delta \times \textrm{ind}(U_{N, \mathbf{r}_1}) \times \textrm{ind}(U_{N,\mathbf{r_2}}),
\end{equation}
(where the $\delta$ here may differ from that in Eq.~\eqref{eq:uuprimeindex}, but will have the same scaling).

\begin{figure}[t]
\centering
\includegraphics[width=.8\linewidth]{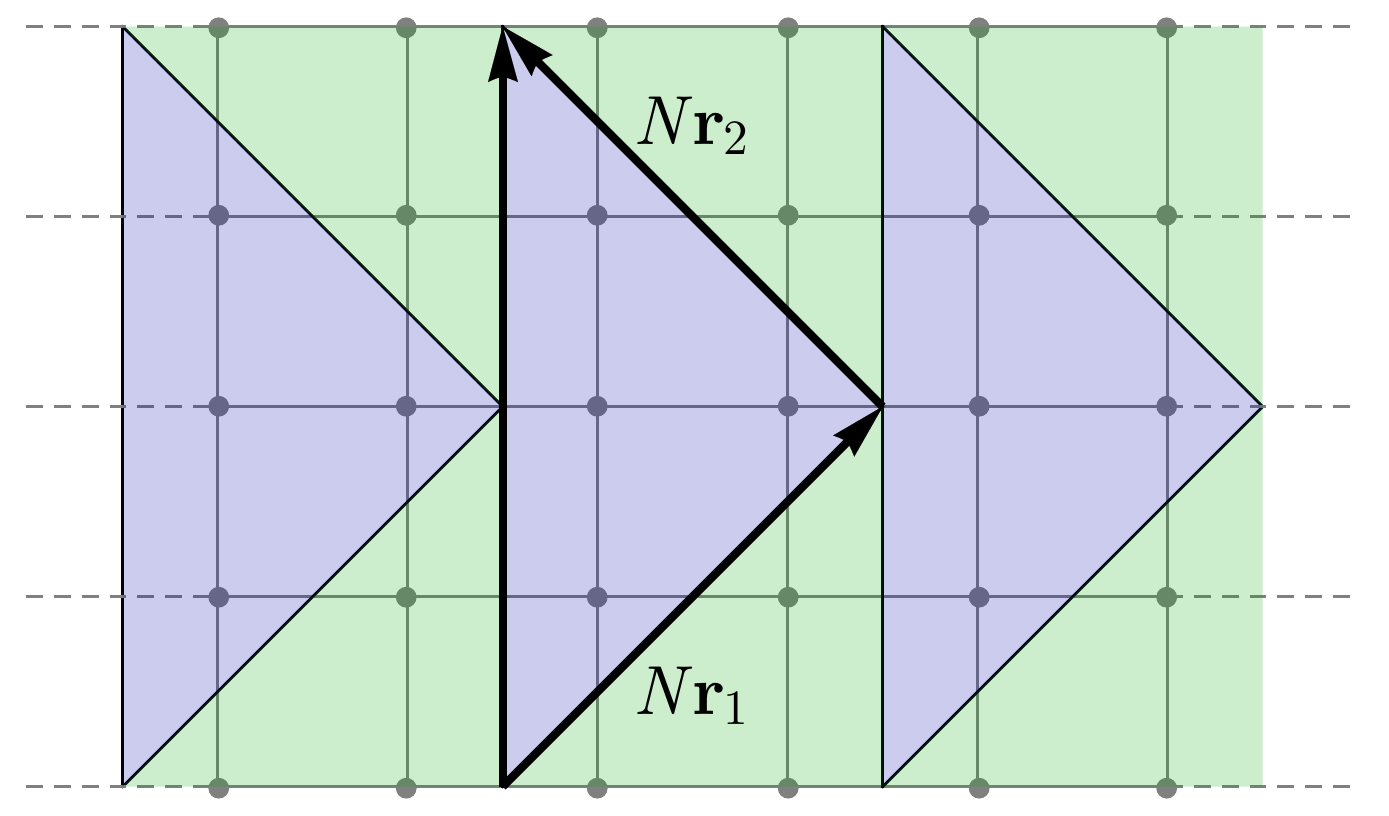}
\caption{Illustration of a 1D cell structure imposed on a 2D system with a compact dimension. For a system with compact dimension of length $N (\mathbf{r}_1 + \mathbf{r}_2)$ we can group the Hilbert spaces of sites within a blue region and pair of green systems into a single site on a 1D chain. The 1D GNVW index, however, is independent of the choice of location of cut used to define $L$ and $R$ in its computation. Therefore, dividing the 2D system using a `triangular' cut (along $N\mathbf{r}_1$ followed by $N\mathbf{r}_2$) or using a `linear' cut (along $N(\mathbf{r}_1 + \mathbf{r}_2)$) gives the same index ${\rm ind}(U)$.}
\label{fig:cutcells}
\end{figure}

The multiplicative correction to the index $\delta$, which is introduced when rejoining two periodic systems, is bounded above and below by constants which depend only on the Lieb-Robinson length $\lambda_{\rm LR}$ of the underlying 2D unitary $U$, and the on-site Hilbert space dimension $d$\footnote{Explicitly, for any 1D system, the largest index is achieved by a unitary whose action is equivalent to Hilbert-space translation by the Lieb-Robinson length. The upper bound on $\delta$ describes the case where the unitary before cutting and rejoining translates a region of dimension $\lambda_{\rm LR}$ near each cut from $L$ to $R$ by a distance $\lambda_{\rm LR}$, but after cutting and rejoining translates the region from $R$ to $L$ by $\lambda_{\rm LR}$. The lower bound is obtained by considering the opposite case.}. Importantly, $\delta$ is essentially independent of the system size $N$, and stays approximately constant as the limit $N\to\infty$ is taken. 

By constructing a sequence of systems with increasing $N$, and using Eqs.~\eqref{eq:defscaledadditiveindex} and~\eqref{eq:multiplicativeindexvectoradd}, we obtain the relation
\begin{equation}
\nu(\mathbf{r}_1 + \mathbf{r}_2) = \nu(\mathbf{r}_1) + \nu(\mathbf{r}_2).
\end{equation}
In particular, we see that $\nu$ is a $\mathbb{Z}$-linear function of 2D lattice vectors. The coefficients $n_p(\br)$ in the sum over primes in Eq.~\eqref{eq:indexdecomp} are therefore integer-valued $\mathbb{Z}$-linear functions of $\br$, and so each may be written as
\begin{equation}
n_p(\mathbf{r}) = \frac{1}{2\pi} \mathbf{G}_p \cdot \mathbf{r},
\end{equation}
given as the inner product of $\br$ with some reciprocal lattice vector $\mathbf{G}_p$. 

Translationally invariant unitaries in two dimensions are therefore completely classified by a set of reciprocal lattice vectors $\{\mathbf{G}_p\}$, indexed by primes $p$. These determine the scaled additive index $\nu(\br)$ along any direction $\br$.  Conversely, by `measuring' $\nu(\br)$ for a unitary $U$ along some basis $\{\mathbf{r}_1,\mathbf{r}_2\}$ of the lattice, we can uniquely determine the vectors $\{\bG_p\}$ using the relation
\beq
\mathbf{G}_p = n_p(\mathbf{r_1}) \mathbf{b}_1 + n_p(\mathbf{r}_2) \mathbf{b}_2,
\eeq
where $\{\bb_1, \bb_2\}$ are reciprocal lattice vectors corresponding to $\{\br_1, \br_2\}$ (satisfying $\br_i\cdot\bb_j=2\pi\delta_{ij}$). Since this classification is discrete, it partitions the set of 2D translationally invariant unitaries into discrete equivalence classes.

We can define a representative unitary $V_{\{\bG_p\}}$ corresponding to a given set of vectors $\{\bG_p\}$ as follows. We first consider the set of translation vectors $\{\br_{{\rm tr},p}\}$, defined by
\beq
\mathbf{r}_{{\rm tr},p} = \frac{1}{2\pi} \left[(\mathbf{r}_1 \times \mathbf{r}_2) \times \mathbf{G}_p\right],\label{eq:rtr_2d}
\eeq
where it may be verified that $\mathbf{r}_{{\rm tr},p}$ is a vector in the direct lattice with basis $\{\mathbf{r}_1,\mathbf{r}_2\}$. For each value of $p$ with a nonzero reciprocal lattice vector $\bG_p$ there is a corresponding nonzero translation vector $\mathbf{r}_{{\rm tr},p}$. For each such value of $p$, we define a local Hilbert space with dimension $p$ on each site; the total Hilbert space is the tensor product of these Hilbert spaces over the complete 2D lattice.

The representative unitary $V_{\{\bG_p\}}$ acts independently on each $p$-dimensional factor of this Hilbert space as a translation with vector $\mathbf{r}_{{\rm tr},p}$. In other words, the unitary $V_{\{\bk_p\}}$ acts as a tensor product of one-dimensional shift operators, but where each factor $\sigma_p$ (corresponding to a different prime value of $p$) may shift in a different direction (and magnitude) $\mathbf{r}_{{\rm tr},p}$. By expressing a given vector $\mathbf{r}$ in the basis $\{\mathbf{r}_1, \mathbf{r}_2\}$ and exploiting the linearity of $\nu(\br)$, it may be verified that this representative unitary $V_{\{\bG_p\}}$ generates the expected value of the chiral unitary index $\nu(\br)$ for any choice of cut $\br$.

The set of reciprocal lattice vectors $\{\mathbf{G}_p\}$ characterizing a particular equivalence class of unitaries inherits a group structure under two products within the space of unitaries from the group structure of the GNVW index \cite{Gross:2012fm}. Under the sequential action of two unitaries $U_3 = U_2 \circ U_1$, the reciprocal lattice vectors add term-wise, $\{\mathbf{G}_{p,3} = \mathbf{G}_{p,1} + \mathbf{G}_{p,2}\}$. Similarly, if we consider the site-wise tensor product of two systems, with unitary $U_3 = U_1 \otimes U_2$, the reciprocal lattice vectors again add term-wise according to $\{\mathbf{G}_{p,3} = \mathbf{G}_{p,1} + \mathbf{G}_{p,2}\}$. In Appendix~\ref{app:further_details_2D} we show that an arbitrary set of translations can always be characterized by a set of reciprocal lattice vectors $\{\bG_p\}$ with $p$ prime. In Appendix~\ref{app:stability}, we show that edge behavior described by different $\bG$ is stable under local (in 2D) unitary deformations at the edge.

%%%%%
\subsection{2D boundaries of 3D unitary loops\label{sec:2D_boundaries_of3D}}
\begin{figure}[t]
\centering
\begin{overpic}[width=.9\linewidth]{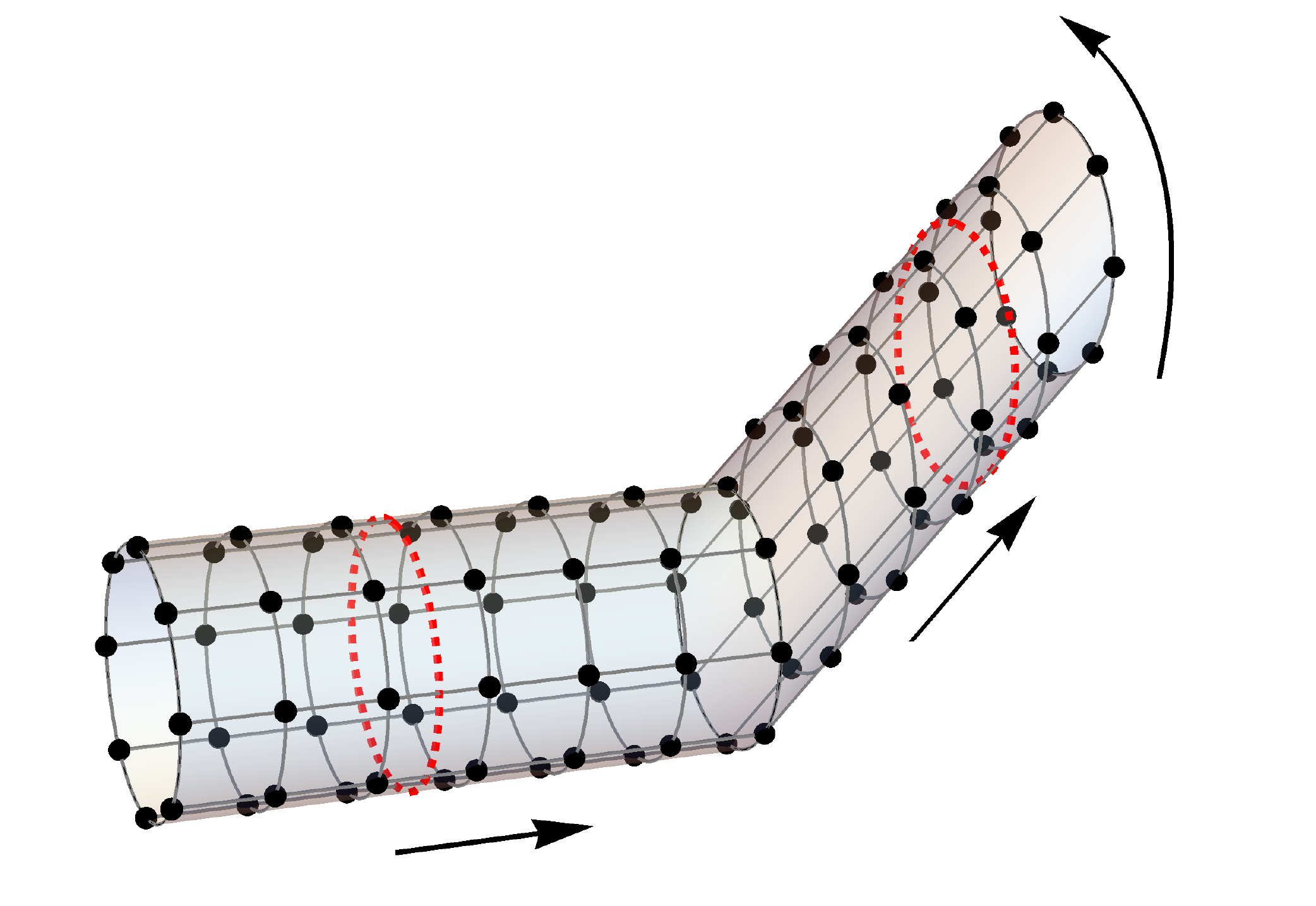}
\put(35,0){$\br_1'$}
\put(74,22){$\br_2'$}
\put(90,55){$\br$}
\end{overpic}
\caption{Illustration of the interface between two periodic systems with shared compact dimension along $\mathbf{r}$, and extended dimensions along $\mathbf{r}'_1$ and $\mathbf{r}'_2$. The chiral unitary index may be calculated by grouping sites along the $\mathbf{r}$ direction and dividing the resulting 1D system into two halves, $L$ and $R$. The dashed red lines show two possible cuts for dividing the system. The chiral unitary index is independent of the location of the cut.}
\label{fig:cylinderinterface}
\end{figure}

Since we are ultimately interested in 3D bulk drives, we now extend our discussion to 2D systems embedded in 3D. We take some translationally invariant 3D unitary loop drive $U_{\rm 3D}$, defined in $\mathbb{R}^3$, which may be used to generate a 2D effective edge unitary at any 2D boundary. If the boundary is a 2D plane, then the surface behavior falls into equivalence classes exactly as described above. To describe the behavior at more complicated boundary surfaces, however, we consider two 2D Bravais lattices $\mathcal{L}_1$ and $\mathcal{L}_2$, which intersect at a common 1D sublattice with primitive lattice vector $\mathbf{r}$. Each lattice $\mathcal{L}_{1/2}$ is spanned by the basis $\{\mathbf{r},\mathbf{r}'_{1/2}\}$. We define the complete boundary system to consist of sites belonging to $\mathcal{L}_1$ on one side of the common sublattice, and sites belonging to $\mathcal{L}_2$ on the other. The underlying bulk drive $U_{\rm 3D}$ is a translationally invariant unitary loop, and so this procedure defines an effective edge unitary $U_{\rm eff}$ that acts on the quasi-2D boundary system.

Since $\mathbf{r}$ is a vector in both $\mathcal{L}_1$ and $\mathcal{L}_2$, we can still define a periodic system by identifying the Hilbert spaces of sites displaced by $N\mathbf{r}$, as illustrated in Fig.~\ref{fig:cylinderinterface}. We can therefore again compute ${\rm ind}(U_{N,\mathbf{r}})$ by dividing the system along $\mathbf{r}$ into two halves, $L$ and $R$. However, the GNVW index is a local invariant \cite{Gross:2012fm}, and so the value of ${\rm ind}(U_{N,\mathbf{r}})$ is independent of the location of the dividing cut. In particular, far from the interface (where the axial dimension is either $\mathbf{r}'_1$ or $\mathbf{r}'_2$), a computation of ${\rm ind}(U_{N,\mathbf{r}})$ will yield the same result. By taking the limit $N\to\infty$, we see that the scaled index $\nu(\br)$ is consistent across the entire boundary.  

\begin{figure}[t]
\centering
\begin{overpic}[width=.5\linewidth]{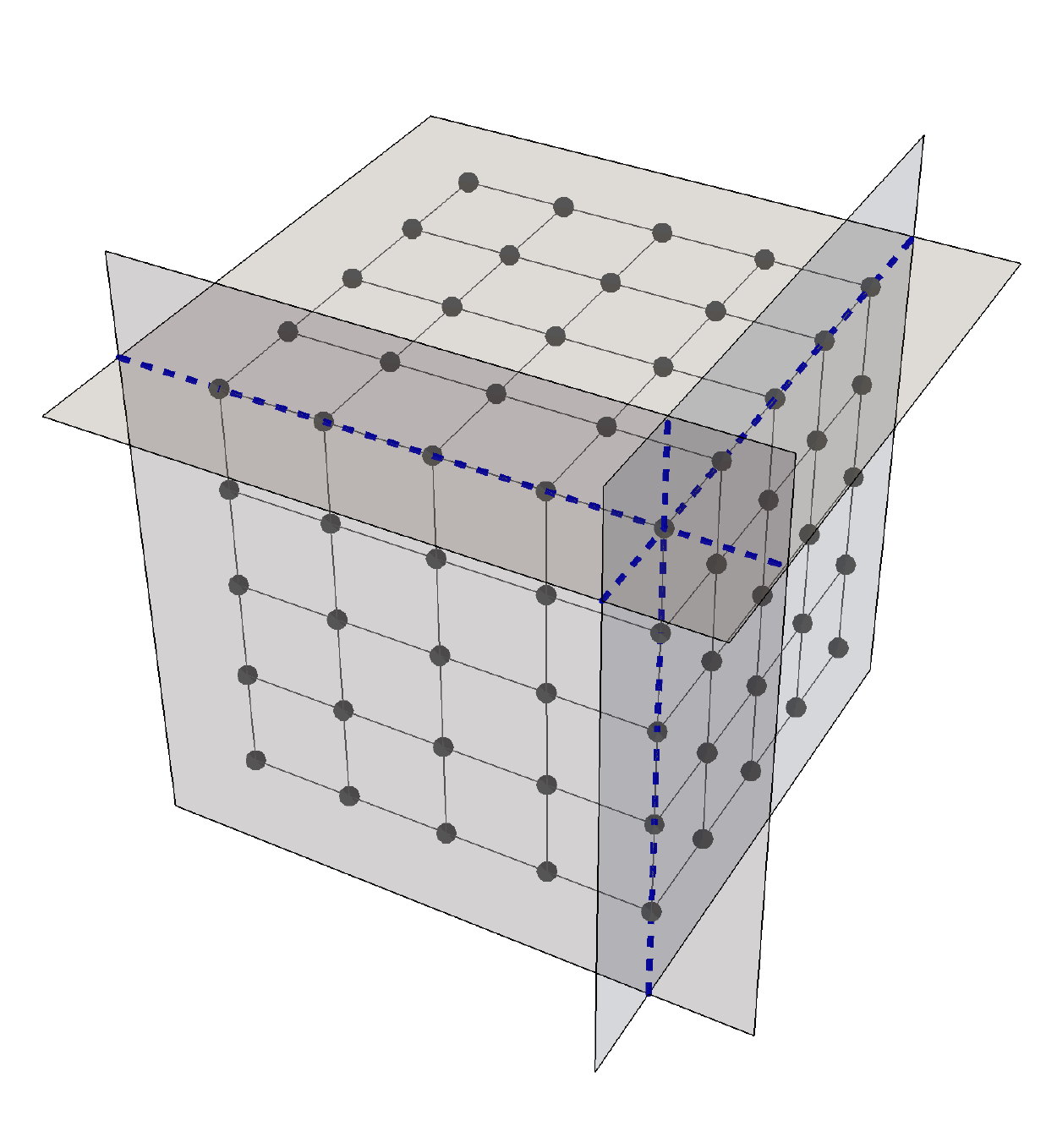}
\end{overpic}
\caption{Three pairwise-intersecting planar boundaries of a 3D system. The intersection between each pair of planes is spanned by a basis vector of the 3D lattice. Values of the chiral unitary index computed within different boundary planes must be consistent with each other and with linearity.}
\label{fig:cubeinterface}
\end{figure}

The arguments above apply to any pair of 2D planar boundaries which intersect at a line. For a 3D bulk unitary $U_{\rm 3D}$, we can find three pairwise-intersecting planar boundaries, in which the interface between each pair is a 1D sublattice spanned by a basis vector of the 3D lattice. This is illustrated in Fig.~\ref{fig:cubeinterface}. Since the scaled additive index is a locally-computed quantity, the values of $\nu(\br)$ computed within different 2D planar boundaries must be consistent with each other and with the linearity described in Eq.~\eqref{eq:multiplicativeindexvectoradd} (where $\br$ is now promoted to a lattice vector in 3D). Overall, this means that the effective edge behavior of a translationally invariant 3D loop drive is fully classified by a set of \emph{three-dimensional} reciprocal lattice vectors $\{\bG_p\}$, indexed by primes $p$. The scaled additive index $\nu(\br)$ is then specified for any 2D boundary and any 1D cut within this boundary (defined by three-dimensional lattice vector $\br$). Effective edge behaviors arising from different 3D bulk unitary loop drives may therefore be put into equivalence classes, each labeled by a set of 3D reciprocal lattice vectors $\{\mathbf{G}_p\}$ (with $p$ prime). In turn, each 3D unitary loop must have an edge behavior belonging to one of these classes, and the space of locally generated 3D loops inherits the classification.

Just as in the 2D case, we can define a representative effective edge unitary $V_{\{\bG_p\}}$ on a particular boundary which corresponds to a given set of vectors $\{\bG_p\}$. As before, we define the set of translation vectors $\{\br_{{\rm tr},p}\}$ through
\begin{equation}
\mathbf{r}_{{\rm tr},p} = \frac{1}{2\pi} \left[(\mathbf{r}_1 \times \mathbf{r}_2) \times \mathbf{G}_p\right],\label{eq:rtr_3d}
\end{equation}
but where $\br_1$, $\br_2$ and $\bG_{p}$ are now 3D vectors. The representative unitary $V_{\{\bG_p\}}$ acts as a translation with vector $\mathbf{r}_{{\rm tr},p}$ on a $p$-dimensional Hilbert space factor on each site. Other effective edge unitaries within the same class must be related to this representative edge unitary by a finite sequence of local 2D unitary evolutions. 

For a given equivalence class and boundary surface, the flow of information per unit cell across a cut in the direction of $\br$ is characterized by the index
\be
\nu(\br)&=&\frac{1}{2\pi}\sum_p\left(\bG_p\cdot\br\right)\log p.
\ee
As an example, Fig.~\ref{fig:surfacecut} shows the action of a simple effective edge unitary and gives the associated vectors $\br_{{\rm tr},p}$ and index $\nu(\br)$ for a choice of cut $\br$.

\begin{figure}[t]
\centering
\includegraphics[clip=true, trim = 40 80 80 20, scale=0.4]{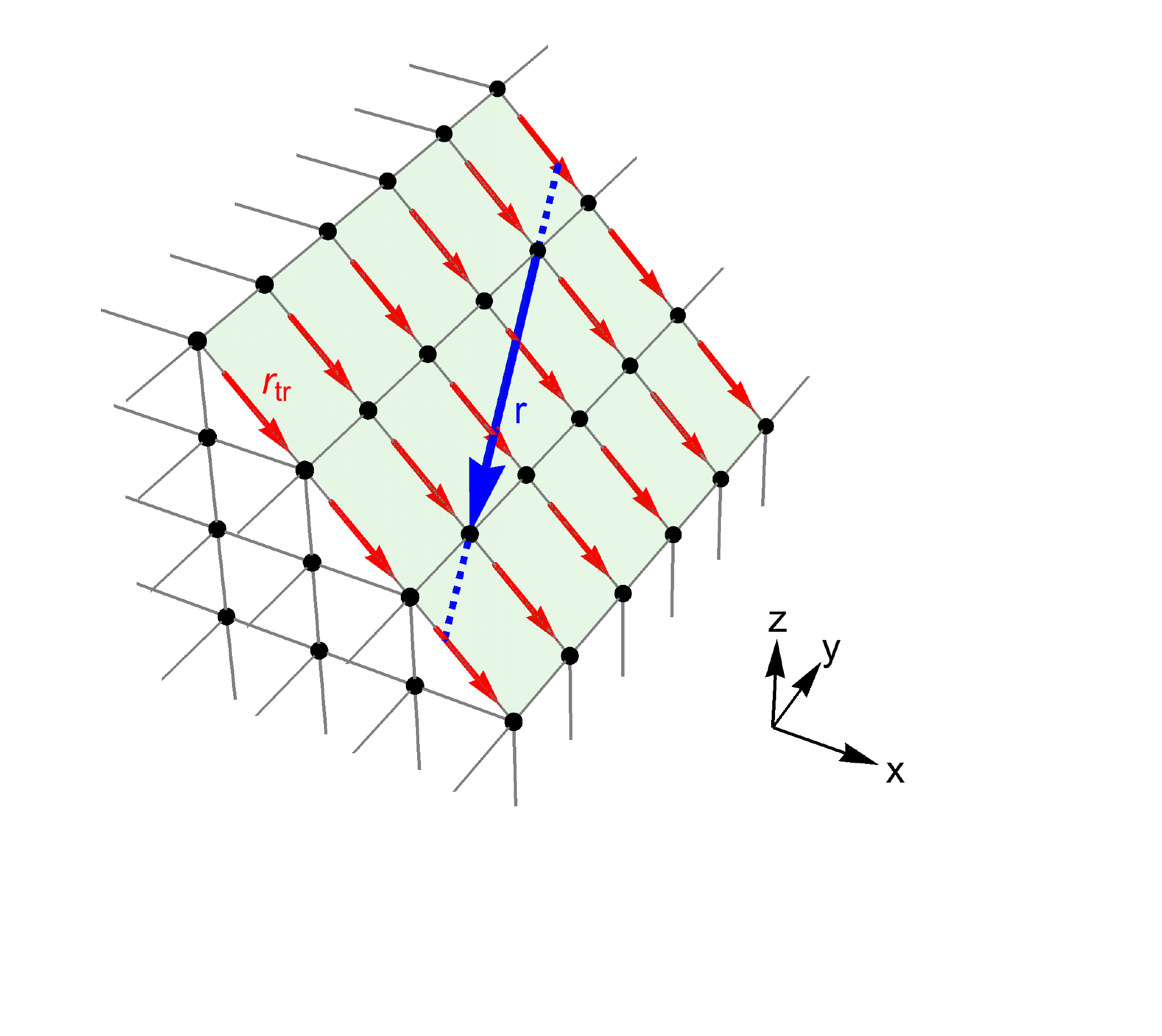}
\caption{The action of a simple effective edge unitary characterized by reciprocal lattice vector $\bG_2=(0,-2\pi,0)$ (with all other $\bG_i$ zero) in a surface with basis $\br_1=(1,0,-1)$ and $\br_2=(0,1,0)$ (note: on-site sublattices are not shown).  Within this surface, the unitary acts as a translation by vector $\br_{\rm tr,2}=(1,0,-1)$, indicated by red arrows. The blue dashed line indicates a 1D sublattice of this surface, with primitive lattice vector $\br=(1,-3,-1)$ indicated by the blue arrow. The flow of information across this cut per sublattice unit cell is quantified by the index $\nu(\br)=1/(2\pi)\left(\bG_2\cdot\br\right)\log 2=3\log 2$. See main text for details.}
\label{fig:surfacecut}
\end{figure}

In the 1D case, each equivalence class of effective 1D edge behaviors has a representative effective edge unitary which is generated by an exactly solvable 2D bulk exchange drive \cite{Po:2016iq,Harper:2017ce}. We will demonstrate that the representative edge unitary of each {two-dimensional} equivalence class may similarly be generated by an exactly solvable 3D bulk exchange drive.

%%%%%
\section{2D Bulk Exchange Drives}\label{sec:2D}
In the previous section, we obtained a classification of local 2D unitary operators with translational invariance, and argued that this provides an equivalent classification of \emph{bulk} Floquet phases in 3D. We showed that each equivalence class is characterized by an infinite set of reciprocal lattice vectors $\{\bG_p\}$, and that each class has a representative effective edge unitary $V_{\{\bG_p\}}$ that is a product of shift operators (or translations) by vectors given in Eq.~\eqref{eq:rtr_3d}. The next aim of this paper is to obtain a set of exactly solvable 3D bulk drives, known as `exchange drives', which may be used to generate these different representative edge behaviors. To aid the discussion, we first review exchange drives in two dimensions and show how they can be used to generate all possible 1D boundary behaviors. In Sec.~\ref{sec:3D}, we will naturally extend these ideas to exchange drives in 3D.

\subsection{Model triangular drive}\label{sec:modeldrive}

We first describe a simple four-step unitary loop drive in 2D which can be used as a building block for more general drives. This is a modification of the models introduced in Refs.~\onlinecite{Po:2016iq,Harper:2017ce}, which in turn build on the noninteracting drive of Ref.~\onlinecite{Rudner:2013bg}.

The model may be defined on any Bravais lattice with a two-site basis. For simplicity, however, we will assume that the lattice is square, has unit lattice spacing, and has both sites within each unit cell (labeled $A$ and $B$) coincident.\footnote{Note that this is in contrast to Refs.~\onlinecite{Rudner:2013bg,Po:2016iq,Harper:2017ce}, in which the lattice basis is nonzero.} On each site of each sublattice there is a finite, $d$-dimensional Hilbert space which, for concreteness, we may assume describes a spin. In this way, the state at a particular site may be written $\ket{\mathbf{r},a,\alpha}$, where $\mathbf{r}$ labels the lattice site, $a \in \{A,B\}$ labels the sublattice, and $\alpha\in\mathcal{H}_{\mathbf{r},a}$ labels the state within the on-site Hilbert space. A basis for many-body states is the tensor product of such states. 

Following Ref.~\onlinecite{Harper:2017ce}, we consider exchange operators of the form
\begin{equation}
U_{\br,\br'}^{\leftrightarrow}=\sum_{\alpha,\beta}\ket{\br,A,\beta}\otimes\ket{\br',B,\alpha}\bra{\br,A,\alpha}\otimes\bra{\br',B,\beta},
\end{equation}
which exchange the state on site $(\mathbf{r}, A)$ with the state on site $(\mathbf{r'},B)$. Note that $U_{\br,\br'}^{\leftrightarrow}$ is local if $\br$ and $\br'$ are nearby, and can therefore be generated by a similarly local Hamiltonian. 

In terms of this operator, we define the four-step drive $U_4 U_3 U_2 U_1$, where each $U_n$ takes the form
\begin{equation} \label{eq:exchangestep}
U_n = \bigotimes_{\br} U_{\br,\br + \bb_n}^{\leftrightarrow},
\end{equation}
with $\mathbf{b}_1 = -(\hat{\mathbf{x}} + \hat{\mathbf{y}})$, $\mathbf{b}_2 = -\hat{\mathbf{y}}$, $\mathbf{b}_3 = \mathbf{0}$, and $\mathbf{b}_4 = -\hat{\mathbf{x}}$. Each step of the drive is a product of exchange operations over disjoint pairs of sites separated by $\mathbf{b}_n$, as illustrated in Fig.~\ref{fig:bulkstepsmodified}(a).

\begin{figure}[t]
\centering
\begin{overpic}
[width=.48\linewidth, clip=true, trim= 0 -10 0 0]{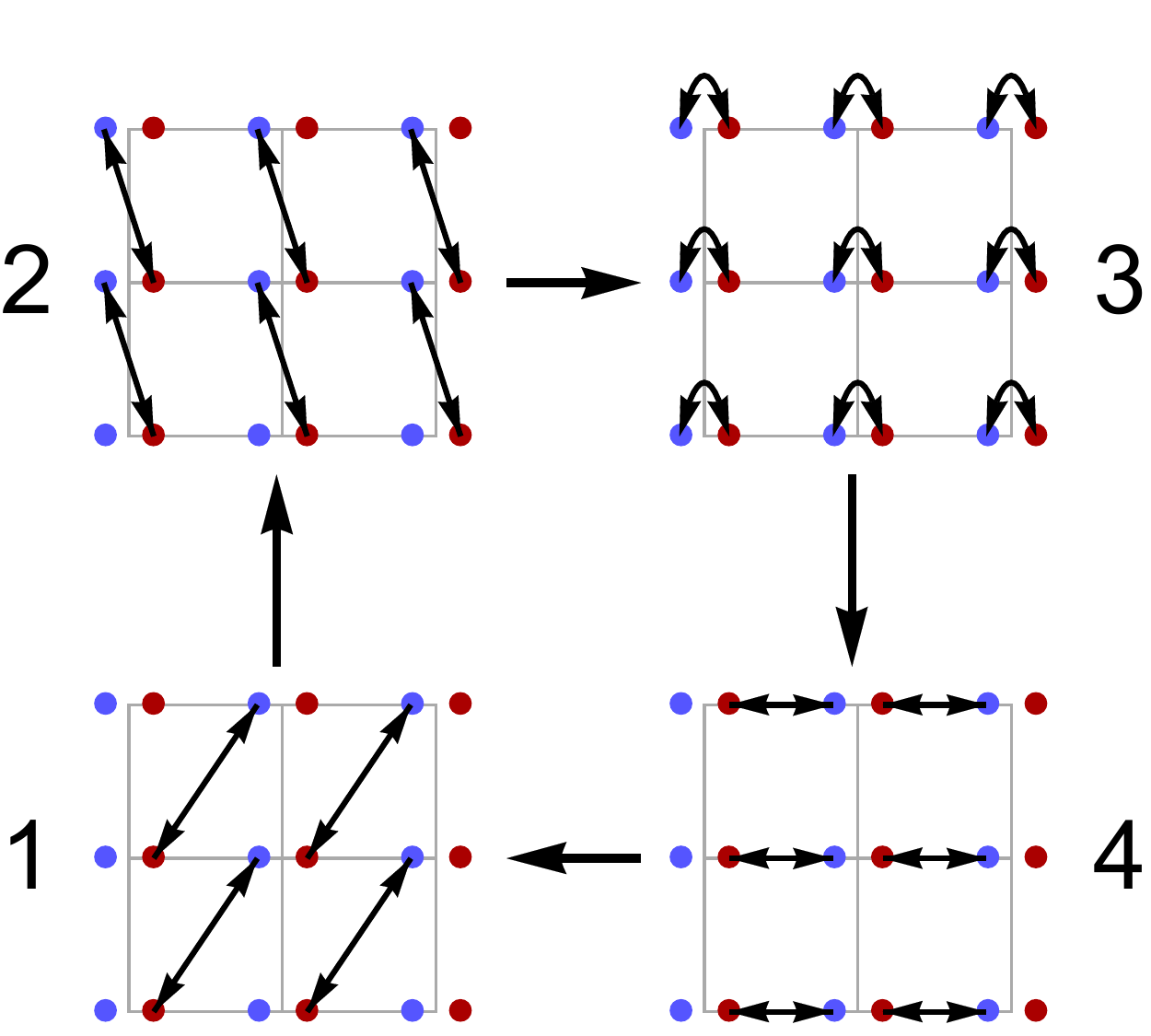}
\put(45,-7){(a)}
\end{overpic}
\hspace{5mm}
\begin{overpic}[width=.4\linewidth,clip=true,trim=0 -70 0 0]{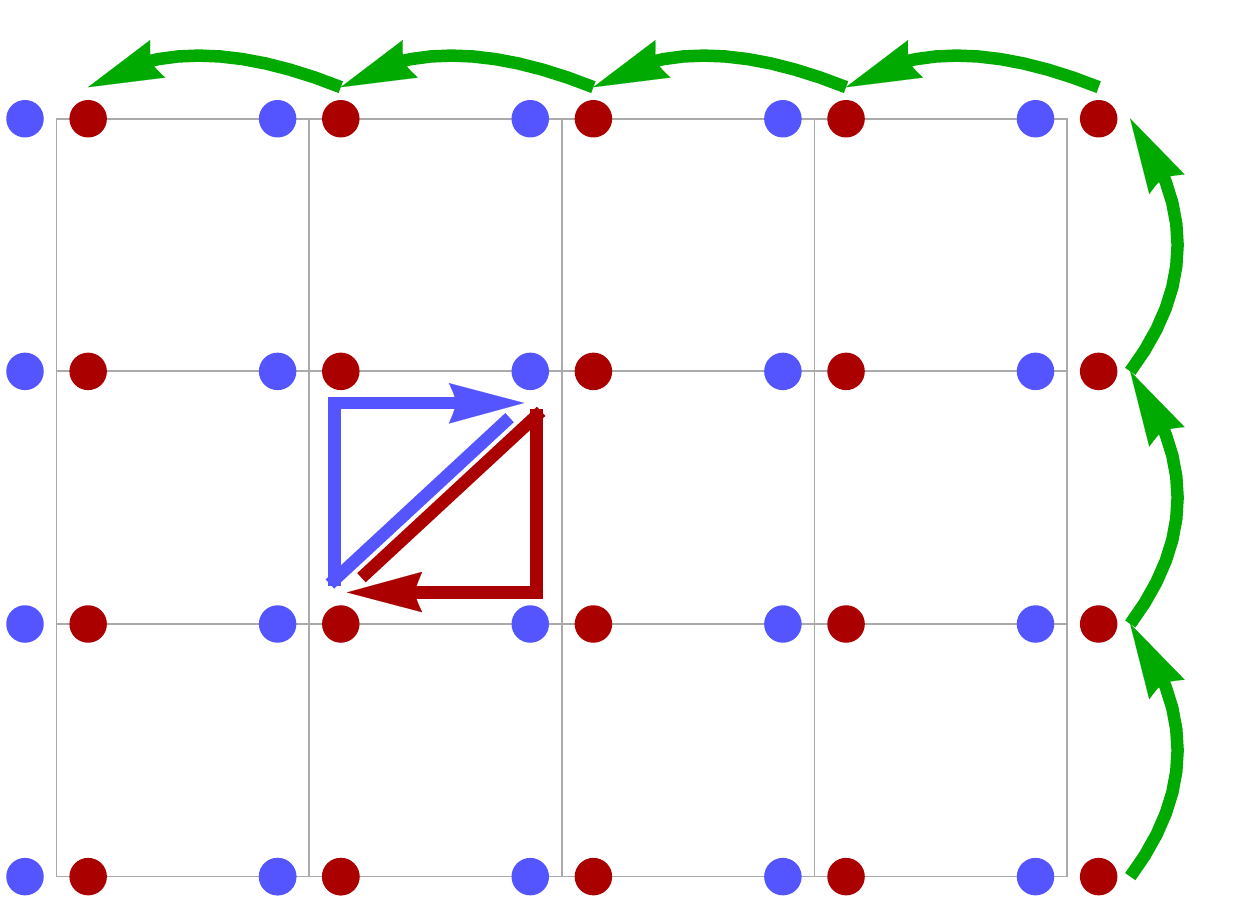}
\put(45,-7){(b)}
\end{overpic}
\caption{Illustration of the four-step exchange drive described in the main text. (a) The steps represent exchanges between nearby on-site states. $A$-sites are depicted in light blue and $B$-sites are depicted in dark red. (b) On-site states in the bulk follow a triangular loop path around a half-plaquette. On-site states at the edge are transported by an effective translation operator represented by the green arrows.}
\label{fig:bulkstepsmodified}
\end{figure}

Since the action of the unitary operator is invariant under lattice translations, we can obtain a complete picture of the drive by focusing on the evolution of a particular on-site component of a generic many-body state. We find that a state beginning at an $A$-site moves in a clockwise loop around the half-plaquette to its lower-left, while a state beginning at a $B$-site moves in a clockwise loop around the half-plaquette to its upper-right, as illustrated in Fig.~\ref{fig:bulkstepsmodified}(b). In this way, each on-site state in the bulk returns to its original position. Since this happens simultaneously for every site, the complete unitary operator acts as the identity on a generic many-body state in the bulk, and is therefore a unitary loop.

At the boundary of an open system, however, some exchange operations are forbidden, and the drive generates anomalous chiral transport~\cite{Po:2016iq,Harper:2017ce}. For the system in Fig.~\ref{fig:bulkstepsmodified}(b), the overall action of the drive is a translation of sublattice states counter-clockwise around the 1D edge: In other words, the effective edge unitary of the drive is a shift $\sigma_d$. By current conservation, this edge behavior must be the same along any edge cut, even if the cut is not parallel to a lattice vector. Note that it would be impossible to generate such a chiral translation with a local Hamiltonian in a purely 1D quantum system~\cite{Harper:2017ce}. 

\subsection{Bulk characterization of 2D exchange drives}
We now construct more general 2D exchange drives from this primitive triangular drive, and show that they may be used to generate all the different 1D edge behaviors (i.e. combinations of shifts) described in Sec.~\ref{sec:edgeclassification2d}.

In the process, we show that the geometry of a generic 2D exchange drive in the bulk is directly related to its edge behavior.

Assuming the same lattice structure as in Sec.~\ref{sec:modeldrive} without loss of generality, we consider a general drive with $2N$ steps, $U = U_{2N} \ldots U_{1}$, with individual steps of the drive being exchanges of the form of Eq.~\eqref{eq:exchangestep}. Each step is characterized by a Bravais lattice vector $\mathbf{b}_n$, which is the displacement between the exchanged sublattice sites directed from $A$ to $B$. After $n$ steps, a state beginning at an $A$-site will be displaced by 
\begin{equation}
\mathbf{d}_n = \sum^{n}_{m = 1} (-1)^{m+1} \mathbf{b}_m,
\end{equation}
where the minus sign arises because each step of the drive moves a state between sublattices. Similarly, a state beginning at a $B$-site will be displaced by $-\mathbf{d}_n$.

Throughout this paper, we are most interested in loop evolutions, which act as the identity in the bulk after a complete driving cycle. The requirement that the drive be a loop enforces the condition
\begin{equation}\label{eq:loopcondition}
\sum^{2N}_{n=1} (-1)^{n+1} \mathbf{b}_n = 0,
\end{equation}
so that the final displacement vector $\mathbf{d}_{2N}$ is zero. We define the signed area of a loop drive by
\begin{equation}\label{eq:signedarea}
A_s = \frac{1}{2 A_{\textrm{prim}}}\sum^{2N-2}_{n = 1} (-1)^{n} \left(\mathbf{d}_n \times \mathbf{b}_{n+1}\right) \cdot \hat{\mathbf{z}},
\end{equation} 
where $\hat{\mathbf{z}}$ is a unit vector perpendicular to the system and $A_{\textrm{prim}}$ is the area of a primitive triangle on the lattice ($A_{\textrm{prim}} = 1/2$ in our convention). Eq.~\eqref{eq:signedarea} calculates the net oriented area enclosed by a state beginning at a site in the bulk and following the complete evolution of the drive, in units of the primitive triangle area. In general, a drive may generate both positively and negatively oriented components, with counter-clockwise loops corresponding to positive areas (see Fig.~\ref{fig:areas}). As defined, the signed area $A_s$ is always an integer, which we will find gives a direct measure of the chiral transport at the edge.

\begin{figure}[t]
\centering
\includegraphics[width=.35\linewidth]{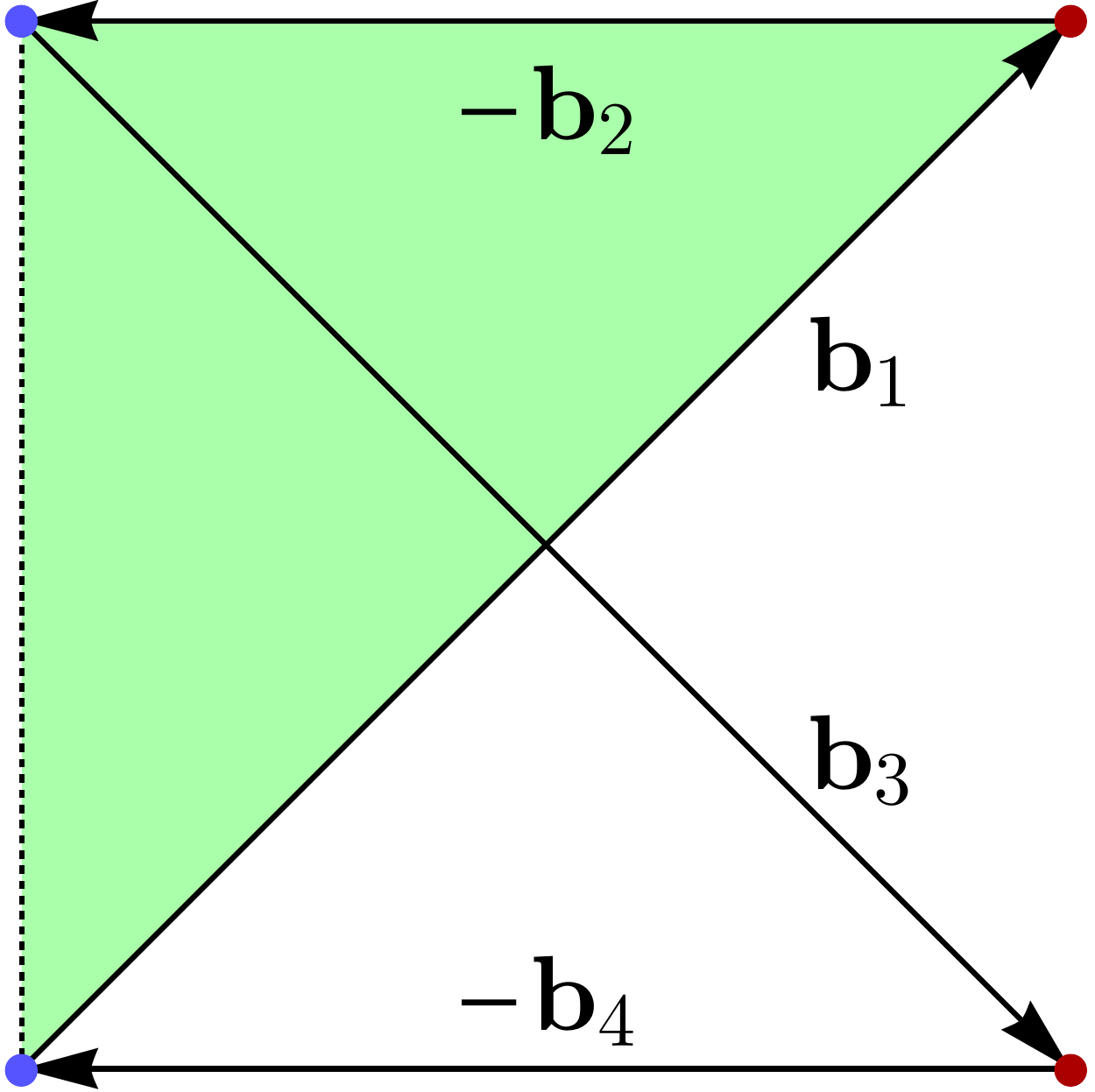}\hspace{8mm}
\includegraphics[width=.35\linewidth]{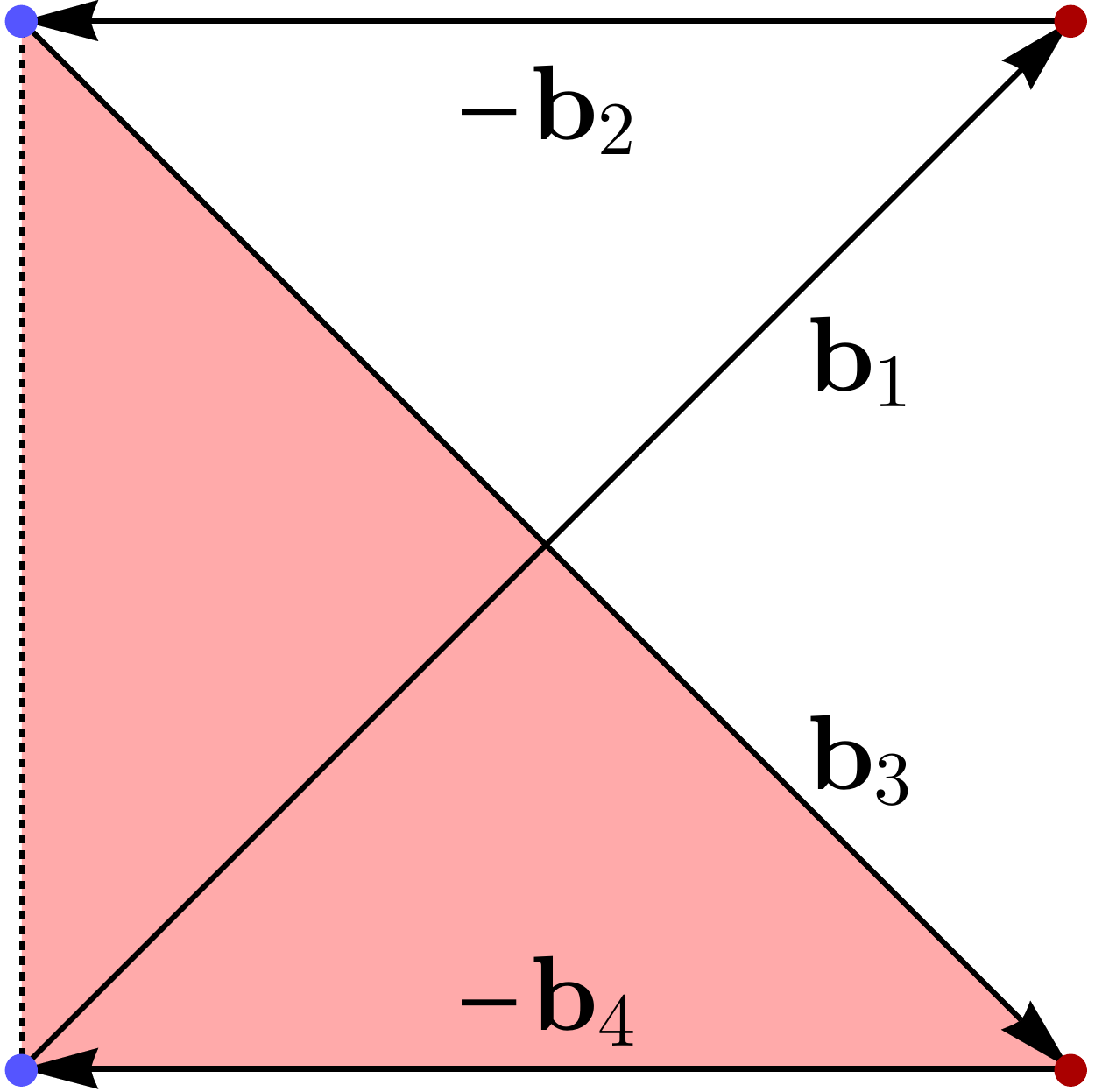}
\caption{Illustration of the signed area summands in Eq.~\eqref{eq:signedarea} for a four-step example drive with $\mathbf{b}_1 = \hat{\mathbf{x}} + \hat{\mathbf{y}}$, $\mathbf{b}_2 = \hat{\mathbf{x}}$, $\mathbf{b}_3 = \hat{\mathbf{x}} - \hat{\mathbf{y}}$, $\mathbf{b}_4 = \hat{\mathbf{x}}$. Since $N = 2$ there are $2N - 2 = 2$ terms in the sum. The signed area of the green (red) triangle represents the first (second) term in the sum and is equal to positive (negative) 1, scaled by the primitive triangle area. In total, this four-step exchange drive has $A_s = 0$.} 
\label{fig:areas}
\end{figure}
 
We now introduce operations that we will use to deform an exchange drive while preserving its signed area and (possibly anomalous) edge behavior. Proofs of these statements may be found in Appendix~\ref{app:modificationproofs}. First, we define a trivial drive to be an exchange drive in which states follow some exchange path and then exactly retrace this path in reverse, satisfying the condition $\mathbf{b}_{n} = \mathbf{b}_{2N - (n-1)}$. The signed area of a trivial drive is zero by construction.

Next, given an exchange drive, we note that we may continuously insert trivial drives at any point without affecting its signed area or edge properties. That is, given a general drive $U = U_{2N} \ldots U_{1}$ and a trivial drive $T$, the drive $U' = U_{2N} \ldots U_{n} T U_{n-1} \ldots U_1$ is continuously connected to $U$. One may also continuously deform an exchange drive by cyclically permuting its steps. These deformations do not affect the signed area of the drive and leave the transport at the edge unaffected, results which are proved in Appendix~\ref{app:modificationproofs}.\footnote{Note that these properties can be demonstrated without appealing to the edge classification discussed in Sec.~\ref{sec:edgeclassification}.}

Using the tools above, we can decompose a general loop exchange drive into a sequence of four-step triangular loop drives. To do this, we insert a trivial drive between each pair of steps that does not include the first or final step. The nature of the trivial drive inserted will depend on the parity of the step: After odd steps, we insert the trivial drive $U'_{2n+1} U'_{2n+1}$, where $U'_{2n+1}$ is an exchange step with $\mathbf{b}'_{2n+1} = \mathbf{d}_{2n+1}$. After even steps, we insert the trivial drive $U_{\rm OS} U'_{2n} U'_{2n} U_{\rm OS}$, where $U'_{2n}$ is an exchange step with $\mathbf{b}'_{2n} = \mathbf{d}_{2n}$ and where $U_{\rm OS}$ is an on-site exchange step with $\mathbf{b} = 0$. The extra swap in the even case acts to effectively transform even steps into odd steps. 

After these insertions, the modified drive can be partitioned into a sequence of $(2N - 2)$ four-step loop drives,
\begin{align}\label{eq:triangulardecomposition}
U' &= \ldots  U'_{4} U_{OS}  U_{4} U'_{3} \cdot  U'_{3} U_{3} U_{OS} U'_{2} \cdot U'_{2} U_{OS} U_2 U_1, \\ 
&= \ldots L_3 \cdot L_2 \cdot L_1,\nonumber
\end{align}
a process which is illustrated in Fig.~\ref{fig:triangledecomp}. It is simple to verify that each four-step loop drive in the partition has a minimum of one on-site swap step, and thus forms either a triangular drive or a trivial drive. Since the operations used to modify the drive preserve the signed area, the signed area of the complete drive may be written in terms of its components as
\begin{equation}
A_s(U) = A_s(U') = \sum_n{A_s(L_n)},
\end{equation}
where we have written $A_s(U)$ for the signed area of loop drive $U$, etc.

\begin{figure}[t]
\centering
\includegraphics[width=.65\linewidth]{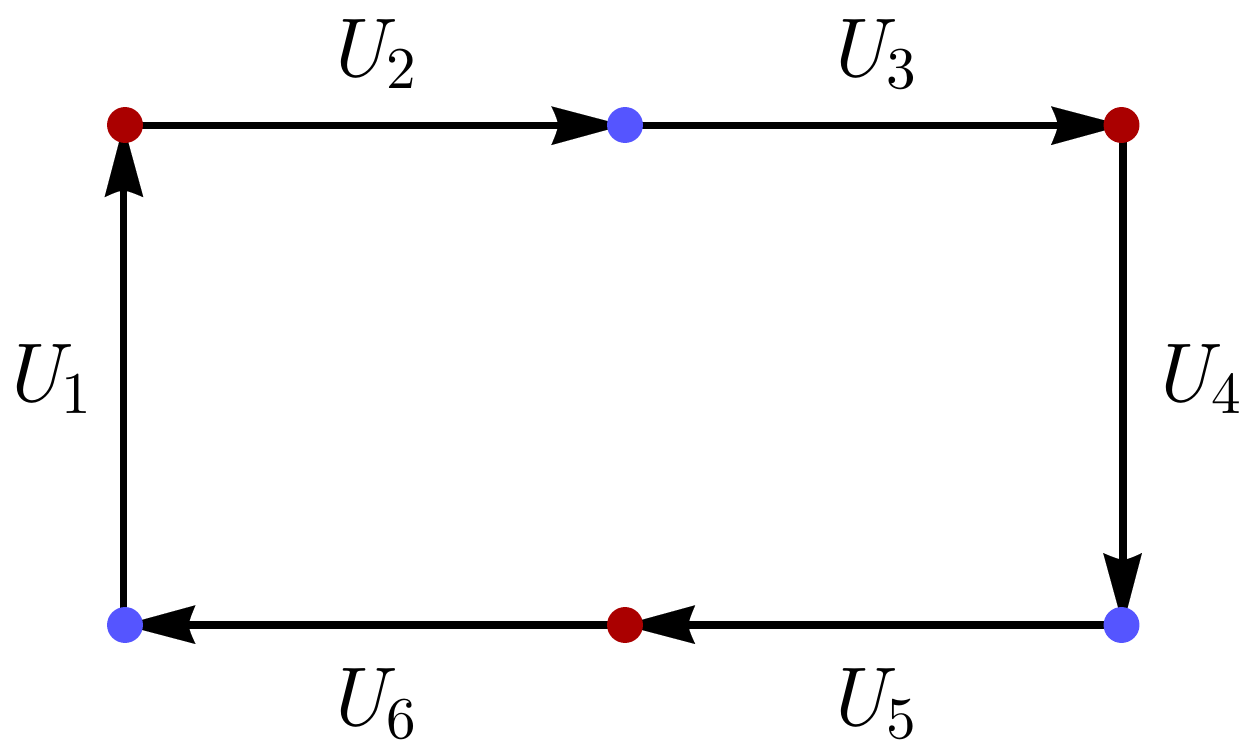}\\
\includegraphics[width=.85\linewidth]{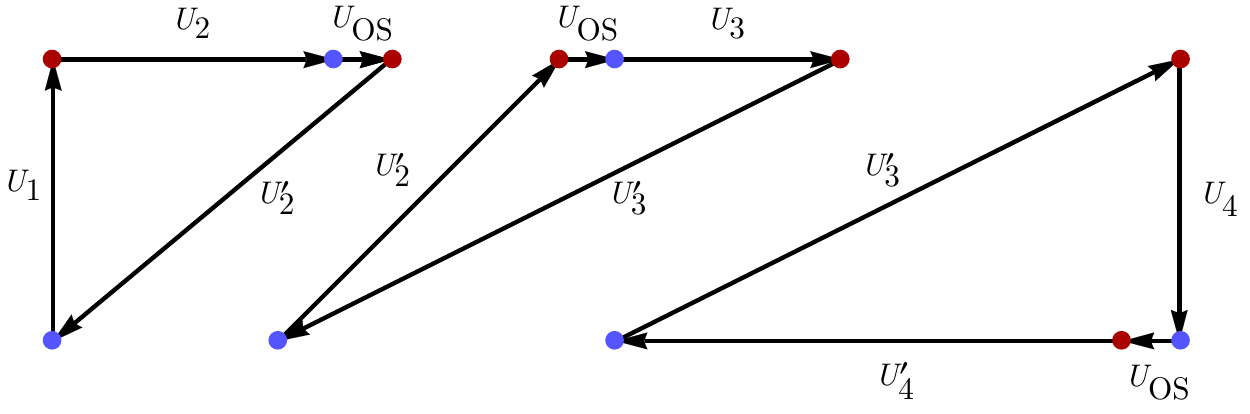}\\
\caption{Illustration of the triangular decomposition in Eq.~\eqref{eq:triangulardecomposition} for an example drive, with steps given by $U_1$ through $U_6$. Since $N = 3$ there are $2N - 2 = 4$ loops in the decomposition but the fourth loop is a trivial drive and we do not depict it here. For clarity, sublattice sites which are not reached by the state localized in the bottom left have been omitted from the figure.}
\label{fig:triangledecomp}
\end{figure}

%%%%
\subsection{Bulk-edge correspondence of 2D exchange drives}
The signed area of a generic drive may be related to its chiral transport at the edge. We define a primitive drive to be a four-step drive in which bulk states follow the path of a primitive triangle, such as the drive described in Sec.~\ref{sec:modeldrive}. Since a primitive drive is triangular, one of its steps must be an on-site swap with $\mathbf{b}_n = 0$. However, as cyclic permutations of loops are equivalent (see Appendix~\ref{app:modificationproofs}), we may assume without loss of generality that the on-site swap occurs on the third step. Therefore, we may equivalently define a primitive drive as a four-step loop drive in which $\{\mathbf{b}_2,\mathbf{b}_4\}$ form a basis for the Bravais lattice and $\mathbf{b}_3 = 0$.

Now, every primitive drive has an effective edge action equivalent either to the model drive in Sec.~\ref{sec:modeldrive} or to its inverse---in other words, its edge action is a shift $\sigma_d$ or a shift $\sigma_{d}^{-1}$. To see this, we perform an invertible orientation- and area-preserving transformation which maps the generic primitive drive (characterized by the basis $\{\mathbf{b}_2,\mathbf{b}_4\}$) onto the model primitive drive presented in Sec.~\ref{sec:modeldrive} or its inverse (characterized by the basis $\{-\by,-\bx\}$ or $\{-\bx,-\by\}$, respectively). The chosen transformation preserves the orientation of sites at the edge, and will map the edge behavior of the generic primitive drive directly onto that of the model primitive drive (or its inverse).

The decomposition of an exchange drive into triangular drives given in Eq.~\eqref{eq:triangulardecomposition} does not generally reduce the original drive to \emph{primitive} drives (as some of the constituent triangles will have areas larger than $A_{\textrm{prim}}$). However, we can use what we know about primitive triangles to deduce the effective edge behavior of a general (nonprimitive) triangular drive, $U_\triangle$. To see this, note that a drive of this form \emph{is} primitive on some number of \emph{sublattices} of the original lattice. This can be shown by considering the sublattice formed from the span of the vectors $\{\mathbf{b}_2,\mathbf{b}_4\}$ defining $U_\triangle$, on which the drive is clearly primitive. Other sublattices on which $U_\triangle$ is primitive can be obtained by translating the first sublattice by the basis vectors of the original lattice. This is illustrated in Appendix~\ref{app:proofarea}, where it is also demonstrated that states on different sublattices do not interact during the drive. 

We claim, and prove in Appendix~\ref{app:proofarea}, that the number of Bravais sublattices $N$ on which a four-step triangular drive is primitive is given by $N = |A_s|$, where $A_s$ is the signed area of that triangular drive. In this way, a four-step triangular drive acts on $|A|$ separate sublattices as either the model drive (if $\sgn(A_s) = -1$) or its inverse (if $\sgn(A_s) = 1$). Since the edge behaviors of the model drive and its inverse are shifts of unit magnitude with opposite chirality, the overall edge behavior of a general triangular drive is $A_s$ copies of the unit shift with the appropriate chirality.
 
Combining the discussions above, we find that the edge behavior of a general 2D translation-invariant exchange drive $U$ is characterized by its signed area in the bulk, $A_s(U)$, and is equivalent to $A_s$ copies of a unit chiral shift. Since the bulk motion of a primitive drive has the opposite chirality to its edge motion, a (negative) positive signed bulk area corresponds to (counter-)clockwise translation at the edge. By forming tensor products of exchange drives, each corresponding to a different on-site Hilbert space, all possible 1D boundary behaviors (with general form $\sigma_p\otimes\sigma_q^{-1}$) can be realized.

%%%%
\section{Bulk and edge behavior of 3D exchange drives}\label{sec:3D}
\subsection{Bulk-edge correspondence for 3D exchange drives}
We now extend the ideas of the previous section to translation-invariant exchange drives in 3D. As in the 2D case, an exchange drive may be defined on any 3D Bravais lattice $\mathcal{L}$ with a two-site basis $\{A,B\}$. For concreteness, we can assume the lattice is cubic and has two coincident sublattices. A boundary of such a system may then be obtained by taking a planar slice through $\mathcal{L}$ to expose some surface containing a 2D Bravais sublattice. As discussed in Sec.~\ref{sec:edgeclassification}, the edge behavior within this boundary can be characterized by the scaled unitary index $\nu(\br)$, defined across a cut in the direction of $\br$.

As before, we consider bulk exchange drives comprising $2N$ steps of the form in Eq.~\eqref{eq:exchangestep}, with each $\mathbf{b}_n \in \mathcal{L}$ now a 3D lattice vector. We recall that these exchange drives are loops, and that they involve local exchange operations that occur throughout the lattice simultaneously (due to translational invariance). Generalizing the signed area of Eq.~\eqref{eq:signedarea}, we claim that the bulk characterization of a 3D drive is given by the reciprocal lattice vector 
\begin{equation}
\bG = \frac{2 \pi}{V_r} \sum^{2N-1}_{n = 1} (-1)^{n} \left(\mathbf{d}_n \times \mathbf{b}_{n+1}\right),\label{eq:reciprocal_lattice_vector}
\end{equation}
where $V_r$ is the volume of the direct lattice unit cell. We will show that this bulk invariant $\bG$ is directly related to the set of reciprocal lattice vectors $\{\bG_p\}$ (introduced in Sec.~\ref{sec:2D_boundaries_of3D}) which characterize the edge behavior.

As in the 2D case, the bulk characterization may be justified by decomposing a general exchange drive into four-step triangular drives. While the decomposition in Eq.~\eqref{eq:triangulardecomposition} continues to hold, the triangular components are now generally not coplanar. Nevertheless, it follows from the arguments of the previous section that the vector $\bG$ for a general drive is the sum of the $\bG$ for each triangular drive in its decomposition. The decomposition therefore preserves the value of $\bG$, and we can understand the edge behavior of a general exchange drive by focusing on its triangular components.

As in 2D, a triangular drive may be defined by the vectors $\{\bb_1,\bb_2,\bb_3,\bb_4\}$, where a cyclic permutation has been chosen so that $\mathbf{b}_3 = \mathbf{0}$.  In this setup, the triangular drive lies in a plane we call the `triangle plane', which includes the vectors $\mathbf{b}_2$ and $\mathbf{b}_4$. We consider the action of this drive on some 2D boundary lattice, spanned by the basis $\{\br_1,\br_2\}$, which defines a `surface plane'. Neglecting the case where the surface plane and triangle plane are parallel (where the edge behavior is trivial), the intersection of these planes is a 1D Bravais sublattice generated by a primitive vector $\mathbf{a}_1 \in \mathcal{L}$. We can therefore choose an ordered basis $\{\mathbf{a}_1,\mathbf{a}_2, \mathbf{a}_3\}$ for $\mathcal{L}$, where $\{\mathbf{a}_1,\mathbf{a}_2\}$ span the triangle plane (and $\ba_3$ is any linearly independent primitive vector).  Note that $\mathbf{b}_2$ and $\mathbf{b}_4$ are not necessarily primitive vectors, and in general $(\mathbf{b}_2 \times \mathbf{b}_4) = A_s (\mathbf{a}_1 \times \mathbf{a}_2)$, where $A_s$ is the signed area discussed previously. According to Eq.~\eqref{eq:reciprocal_lattice_vector}, this triangular drive will have the characteristic reciprocal lattice vector
\begin{equation}
\bG = \frac{2 \pi}{V_r} (\mathbf{b}_2 \times \mathbf{b}_4) = \frac{2 \pi}{V_r} A_s (\mathbf{a}_1 \times \mathbf{a}_2).
\end{equation}

We now consider the edge behavior of this drive in the surface plane. We can write the ordered basis for the surface plane $\{\mathbf{r}_1, \mathbf{r}_2\}$ in terms of the basis of the 3D lattice as $\mathbf{r_1} = \mathbf{a}_1$ and $\mathbf{r_2} = D\mathbf{a}_2 + E\mathbf{a}_3$ (where $ D,E \in \mathbb{Z}$ are coprime). This surface is equivalently characterized by the outward-pointing reciprocal lattice vector
\begin{equation}
\mathbf{k}_s = \frac{2\pi}{V_r} (\mathbf{r}_1 \times \mathbf{r}_2).
\end{equation}
We claim that the edge behavior of the bulk triangular drive described above is a shift (or translation) within the surface lattice given by the direct lattice vector
\begin{equation} \label{eq:bulktranslationcalculation}
\mathbf{r}_{\textrm{tr}} = \frac{2\pi}{V_k} \left(\bk_s \times \bG \right)=\frac{1}{2\pi}\left[\left(\br_1\times\br_2\right)\times\bG\right],
\end{equation}
where $V_k$ is the volume of the 3D reciprocal lattice unit cell. For the triangular drive above this reduces to
\begin{align}
\mathbf{r}_{\textrm{tr}} &= \frac{1}{V_r}  (\mathbf{r}_1 \times \mathbf{r}_2) \times (\mathbf{b}_2 \times \mathbf{b}_4) = -A_s E  \mathbf{a}_1.
\end{align}

The fact that this is the correct edge behavior can be justified as follows: Since a triangular drive in 3D acts on a stack of parallel decoupled planes, the edge surface will host a 1D shift (or translation) for each triangle plane that terminates on it. The number of triangle planes terminating per unit cell of the 2D boundary sublattice is exactly $E$, and the factor of $A_s$ accounts for the fact that the triangular drive may not be primitive. The overall minus sign arises because the chirality of bulk motion is opposite that of edge motion. Thus, $\mathbf{r}_{\textrm{tr}}$ gives the effective edge translation correctly for a triangular drive and an arbitrary edge surface.

Since $\bG$ for a general exchange drive is given by the sum of $\bG$ over its triangular components, it follows that Eq.~\eqref{eq:bulktranslationcalculation} holds for \emph{any} 3D exchange drive. In this way, Eqs.~\eqref{eq:reciprocal_lattice_vector}~and~\eqref{eq:bulktranslationcalculation} completely characterize the bulk and edge behavior of a generic 3D translation-invariant exchange drive.

%%%%%%
\subsection{Products of 3D exchange drives}\label{sec:general_3D}
In Sec.~\ref{sec:edgeclassification} we found that 2D boundary behaviors form equivalence classes characterized by a set of reciprocal lattice vectors $\{\bG_p\}$. The representative edge behavior of given class is a product of translations by vectors $\br_{{\rm tr},p}$ (defined in Eq.~\eqref{eq:rtr_3d}), each acting on an on-site Hilbert space with prime dimension $p$. In order to generate the edge behavior of a general equivalence class, we should take a tensor product of the bulk exchange drives described above.

For the equivalence class with reciprocal lattice vectors $\{\bG_p\}$, we take a tensor product Hilbert space which has an on-site factor of dimension $p$ for each non-zero $\bG_p$. For each $p$-dimensional subspace, we choose a bulk exchange drive that is characterized by the reciprocal lattice vector $\bG=\bG_p$, as defined in Eq.~\eqref{eq:reciprocal_lattice_vector}. Any bulk exchange drive with this property is suitable, but for simplicity we can always choose a four-step triangular drive with the appropriate area. Then, by the reasoning above, the complete product drive will produce the required translation by lattice vector $\br_{\rm tr,p}$ for each $p$-dimensional subspace on an exposed surface. In other words, a product drive of this form in the bulk will reproduce the representative effective edge unitary of the equivalence class $V_{\{\bG_p\}}$ on an exposed boundary. In this way, 3D product drives of this form are representatives of the different equivalence classes of 3D dynamical Floquet phases.

%%%%%%
\section{Conclusion\label{sec:Discussion}}

In summary, we have studied 3D many-body Floquet topological phases with translational invariance but no other symmetry from the perspective of their edge behavior. We found that phases of this form fall into equivalence classes that are somewhat analogous to weak noninteracting topological phases. Members of each class share the same anomalous information transport at a 2D boundary, which is equivalent to a tensor product of shifts (or translations). The representative edge behavior in each equivalence class can be generated by an exactly solvable exchange drive in the bulk.

These equivalence classes capture all possible topological phases of this form whose edge behavior is equivalent to that of a tensor product of lower dimensional phases. To form a complete classification, however, there would need to exist no intrinsically 3D (`strong') Floquet topological phases (without symmetry). We expect this requirement to hold for the following reason: In 2D, the exchange drives which exhaust the possible Floquet topological phases in class A can be regarded as generalizations of the noninteracting system in Ref.~\onlinecite{Rudner:2013bg}. For an intrinsically 3D phase to exist in the interacting case, we would also expect it to have a similar noninteracting counterpart. However, in Ref.~\onlinecite{Roy:2017cs} it is shown that noninteracting Floquet systems in class A host only a trivial 3D phase. In this way, we conjecture that the classification is complete.

In classifying these phases, we developed a method for determining the effective edge behavior of an arbitrary exchange drive in 2D or 3D using geometric aspects of its action in the bulk. We found that 3D exchange drives may be characterized by an infinite set of reciprocal lattice vectors $\{\bG_{p}\}$, with $p$ indexing prime Hilbert space dimensions. These vectors may be calculated directly from the form of the bulk exchange drive, and completely characterize its edge behavior. The vectors $\{\bG_{p}\}$ share some similarities to weak invariants of static topological insulators~\cite{Halperin:1987gx,Fu:2007io,Fu:2007ei,Roy:2009kk,Roy:2010il}. However, in contrast to the static case, these 3D chiral Floquet phases cannot generally be viewed as stacks of decoupled 2D layers, since different Hilbert space factors within a tensor product may stack in different directions.

Our classification suggests a number of interesting directions for future work. A natural follow-up is to ask whether a similar classification can be obtained for 3D Floquet phases of fermions, as well as in systems with additional symmetries. In addition, by combining these phases with topological order, it may be possible to obtain analogues of the Floquet enriched topological phases found in Refs.~\onlinecite{Po:2017ex,Fidkowski:2017us}. Finally, it would be useful to obtain a rigorous proof of the conjecture that there are no inherently 3D Floquet topological phases in systems without symmetry, perhaps by developing an extension of the GNVW index to higher dimensions~\cite{Gross:2012fm}. 

\begin{acknowledgments}
We thank X.~Liu for useful discussions.  D. R., F. H., and R. R. acknowledge support from the NSF under CAREER DMR-1455368 and the Alfred P. Sloan foundation.
\end{acknowledgments}

\appendix
%%%%
\section{Further details on the classification of 2D effective edge unitaries\label{app:further_details_2D}}
In the main text, we argued that translationally invariant unitary operators in 2D form equivalence classes labeled by a set of reciprocal lattice vectors $\{\bG_p\}$ with $p$ prime. In this appendix, we show that generic (site-by-site) tensor products of such unitary operators always reduce to this form.
 
We first note that we can associate a reciprocal lattice vector $\bG_n$ with each term of such a tensor product, using the arguments of Sec.~\ref{sec:edgeclassification}. We can therefore initially characterize a general product drive by a set of pairs $\{(\bG_n,d_n)\}$, where $d_n$ labels the Hilbert space dimension of the $n$th term (but where the $d_n$ will not generally be prime or unique). To remove any repetition, if any two terms in the product have the same Hilbert space dimension $d_n=d_m$, we may replace the pairs $(\bG_{n},d_n)$ and $(\bG_{m},d_m)$ with the single pair $(\bG_{n}+\bG_{m},d_n=d_m)$. This is because the information transported is equivalent after the replacement, as may be demonstrated by regrouping the sites on the lattice using the methods of Ref.~\onlinecite{Harper:2017ce}. 
To reduce all the Hilbert space dimensions to primes, we may view any term for which $d_n$ is not prime as a tensor product of drives, according to its prime factorization. Explicitly, if $d_n$ = $2^{n_2} 3^{n_3} 5^{n_5} \ldots$, we can replace $\left(\bG_{n},d_n\right)$ with a term for every prime factor $\{(n_2 \bG_{n},2), (n_3 \bG_{n},3), (n_5 \bG_{n},5), \ldots \}$. Again, the information transported in the 2D boundary system is equivalent in both cases.

By performing this reduction to prime dimensions and further combining terms of the same dimension, we find that a general effective edge unitary can always be characterized by a set of reciprocal lattice vectors $\{\bG_{p}\}$, each corresponding to an on-site Hilbert space with prime dimension $p$. Using Eq.~\eqref{eq:indexdecomp}, the scaled chiral flow associated with this effective edge unitary can easily be calculated.

%%%
\section{Stability of 2D effective edge unitaries} \label{app:stability}
In Ref.~\onlinecite{Harper:2017ce} it is shown that a shift (translation) operator $\left(\sigma_p\right)^n$ acting on a 1D boundary cannot be continuously deformed to a different shift operator $\left(\sigma_p\right)^{n'}$ with $n\neq n'$ through a local unitary evolution restricted to the 1D system. This includes the trivial shift operator $\left(\sigma_p\right)^0=\mathbb{I}$. In this appendix we formally show that this stability continues to hold when applied to the more complicated boundary behavior (described by some reciprocal lattice vector $\bG$) that may act at a 2D boundary.

\begin{figure}[t!]
\centering
\includegraphics[scale=0.3]{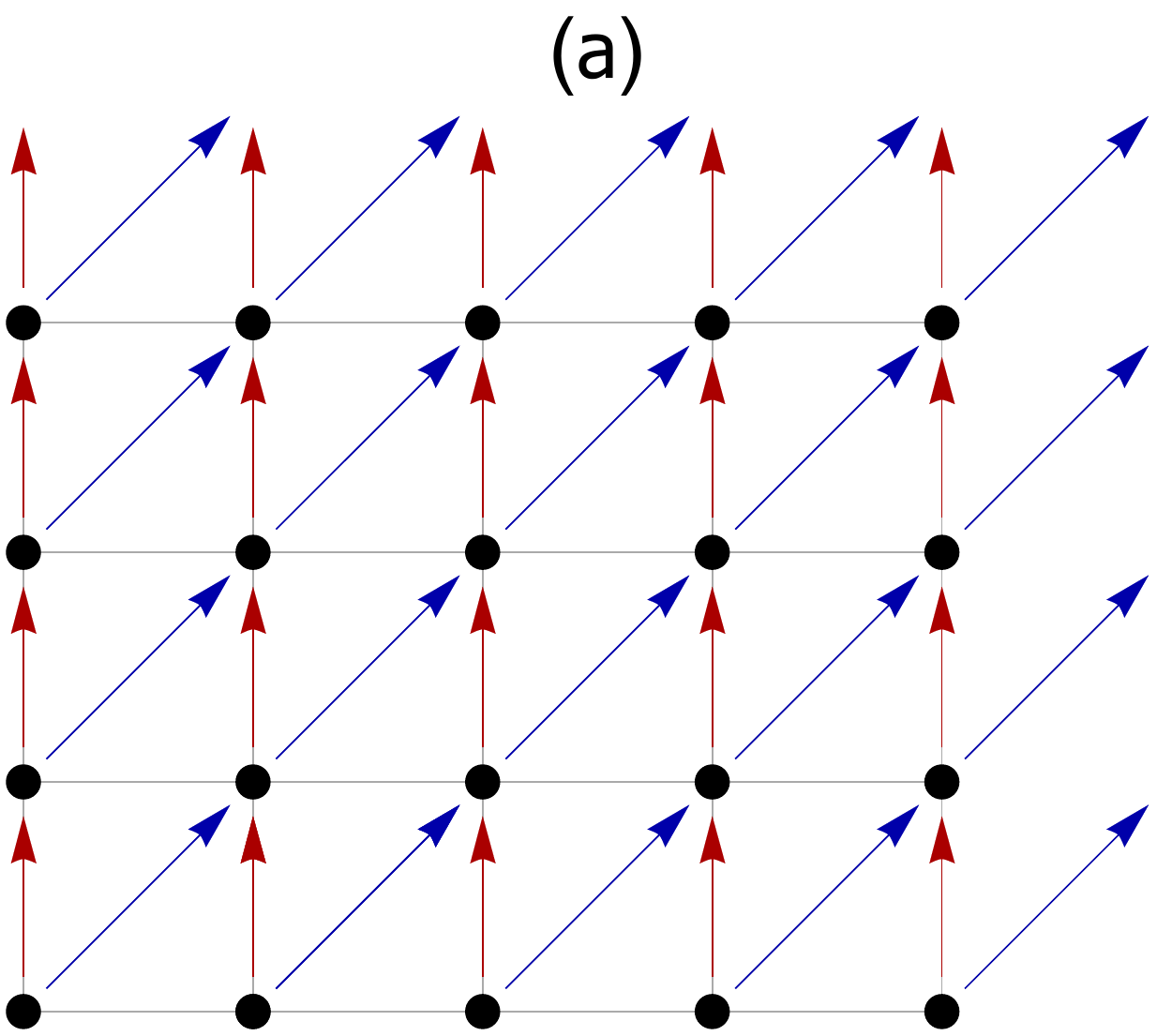}\hspace{5mm}
\includegraphics[scale=0.3]{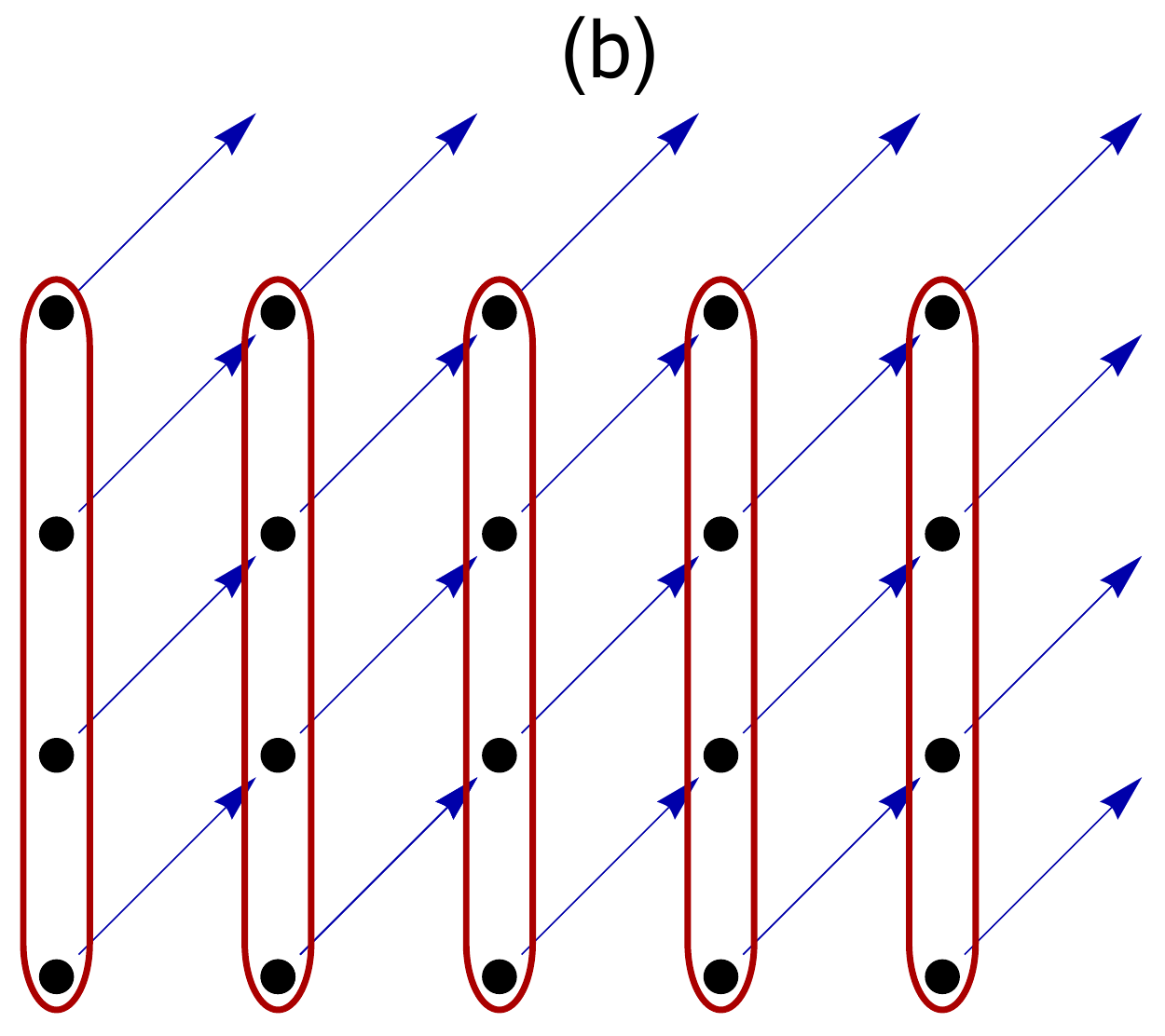}
\caption{Boundary behavior on a 2D surface described by different $\bG$ and $\bG'$ cannot be deformed into one another by local unitary transformations within the boundary. (a) Two different boundary behaviors, corresponding to distinct translation vectors $\mathbf{r}_{\textrm{tr}}$, are indicated (by red and blue arrows) on a 2D boundary. (b) This 2D behavior can be reduced to an effective 1D model by grouping lattice sites in the direction of one of the $\mathbf{r}_{\textrm{tr}}$. In this effective model, one effective edge unitary becomes a permutation of the on-site Hilbert space (within the red grouping) and the other becomes a translation in the horizontal direction combined with a permutation (blue arrows).} \label{fig:appC}
\end{figure}

We consider two 2D boundary systems (which we assume to be identical 2-tori with finite size) with the same on-site Hilbert space dimension $d$. [If these drives have different on-site Hilbert space dimensions or different sizes then they are trivially inequivalent.] On each system, we take unitaries characterized by inequivalent $\bG$ and $\bG'$, leading to distinct behavior. The action of each unitary is characterized by a translation vector within the 2D boundary surface, as argued in Sec.~\ref{sec:edgeclassification}. 

We now create an effective 1D system by grouping the sites on the 2-torus surface as illustrated in Fig.~\ref{fig:appC}. If the translation vectors of the two drives are not parallel, we group together the sites on the 2-torus that lie in the direction of the translation vector of (say) the second drive. If the translations of the two drives are parallel, we group together the sites of the 2-torus that lie along any chosen direction that is not parallel to the translation vectors. In both cases, we are left with two effective 1D edge behaviors that are topologically distinct~\cite{Po:2016iq,Harper:2017ce}. By the arguments of Ref.~\onlinecite{Harper:2017ce}, the two effective edge unitaries cannot be deformed into one another by a local 1D perturbation. This argument holds for each step in the sequence of boundary systems as their size is made infinite.

%%%%
\section{Continuous modifications of loop drives}\label{app:modificationproofs}
In this appendix we define transformations which may be carried out on a unitary exchange drive, and prove that these transformations leave the effective edge behavior unaltered.

\begin{prop}\label{prop:conj}
Given a 2D unitary loop $L$ which acts trivially in the bulk but nontrivially (i.e. as a shift) at the boundary of an open system and a unitary swap $U$ which interchanges pairs of states separated by a finite distance, we consider the sequence of drives $U^{-1} L U$. We claim that this sequence has the same edge behavior as $L$. 
\end{prop}
\begin{proof}
Since $U$ acts as a product over disjoint pairs of sites, we can disentangle its effects in the bulk from its effects on the edge. To do this, we extend the original edge region of $L$ to include sites which are connected to it by the action of $U$. In this way, we can write the composite unitary as the product of the identity in the bulk and a piece which acts at the edge, as shown in Fig.~\ref{fig:appA}. Now, considering the action restricted to this new edge region, the unitary acts as a product of local unitaries and a shift (translation) operator. However, no local 1D unitary evolution can generate (or destroy) chiral edge behavior~\cite{Harper:2017ce}, and so the conjugation with $U$ can have no effect on the chiral properties of $L$.

\begin{figure}[t]
\centering
\includegraphics[scale=0.3]{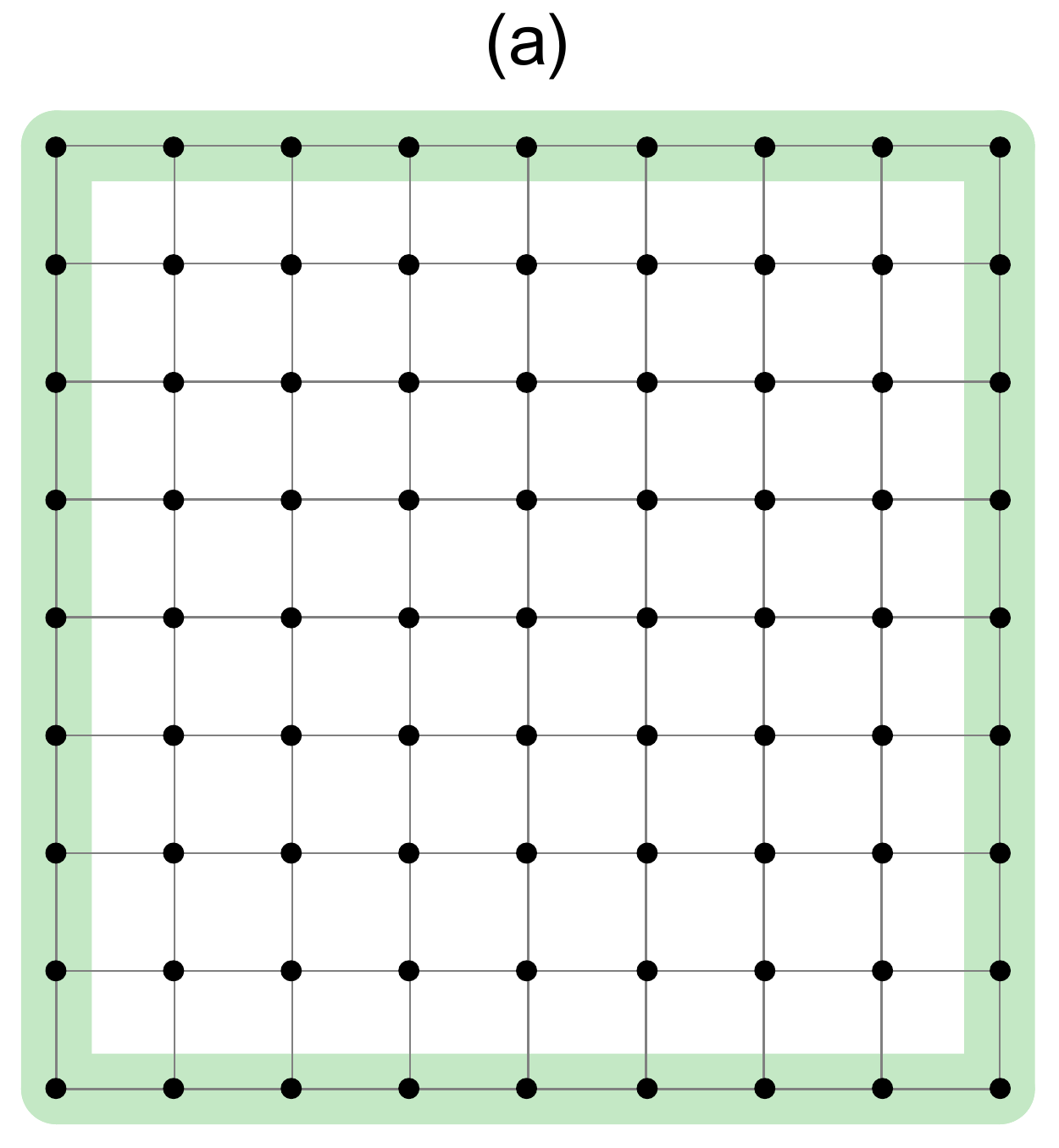}\hspace{5mm}
\includegraphics[scale=0.3]{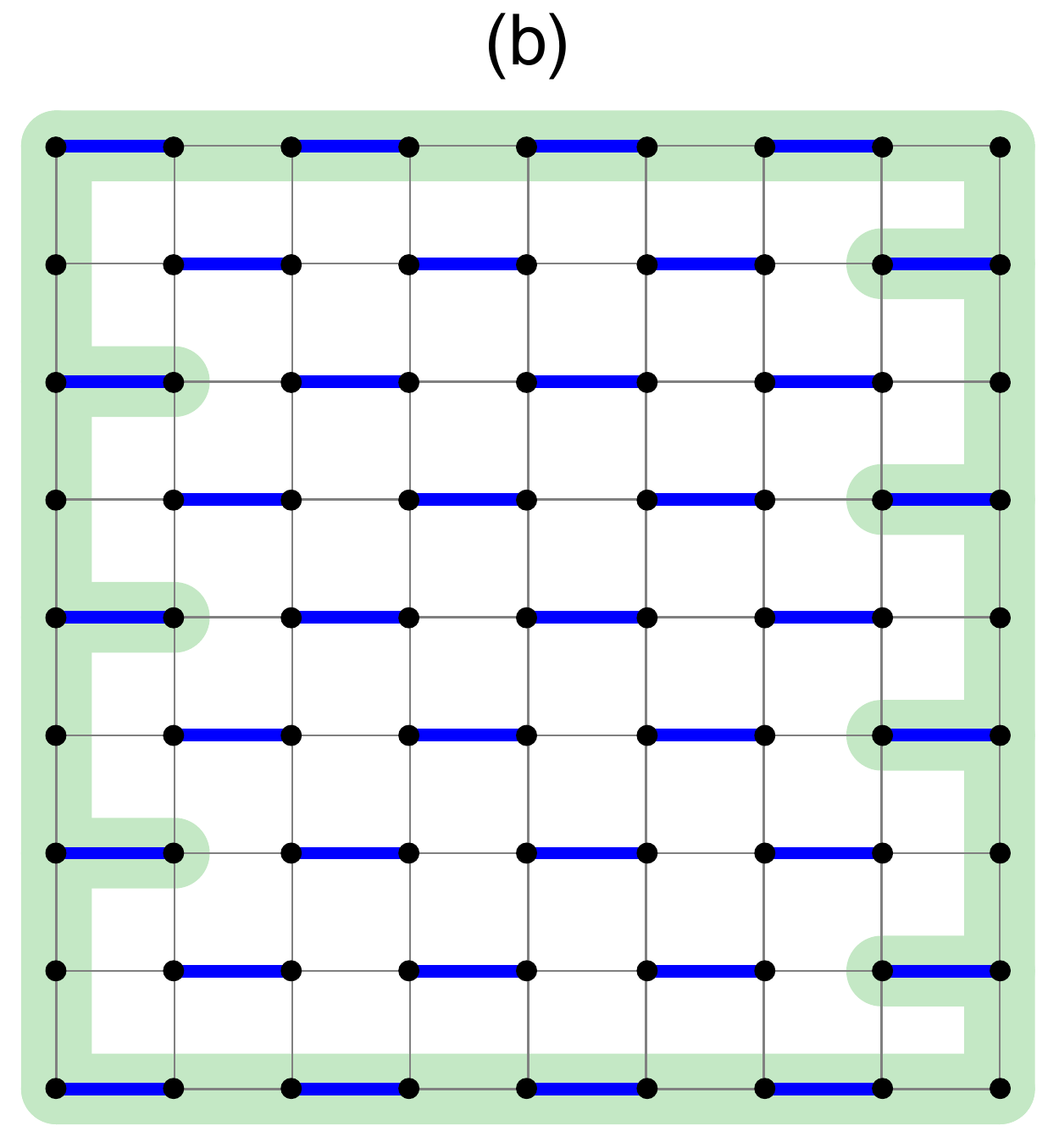}
\caption{(a) A unitary loop $L$ acts trivially in the bulk but may act nontrivially in a quasi-1D edge region located near a boundary (green shaded region). (b) Conjugating the unitary loop $L$ with a product of disjoint pairwise swaps (thick blue lines) may connect bulk sites to the edge region. We define a new quasi-1D edge region which includes these former bulk sites (green shaded region). See main text for details.}
\label{fig:appA}
\end{figure}

An alternative point of view is that conjugation with $U$ acts as a local basis transformation of the Hilbert space restricted to the edge. A local basis transformation of a quasi-1D system cannot change the global properties of the drive.

\end{proof}

Note that $U$ is an exchange operator and can be continuously connected to the identity, and so we can define $U(\theta)$ such that $U(0) = \mathbb{I}$ and $U(1) = U$. We therefore see that conjugation with $U(\theta)$ defines a continuous transformation within the space of unitary loops. Further, note that this composite unitary $U^{-1} L U$ is also a loop as it is trivial in the bulk.

\begin{prop}\label{prop:nconj}
Given a unitary loop $L$ and a finite sequence of local unitary swaps $\{U_1, \ldots, U_N\}$, then the composite unitary operator $(U_1 \ldots U_N)^{-1} L (U_1 \ldots U_N)$ has the same edge behavior as $L$.
\end{prop}
\begin{proof}
One repeats the argument in Proposition~\ref{prop:conj} $N$ times.
\end{proof}

\begin{prop}\label{prop:trivial}
Any drive $T$ comprising a sequence of unitary swaps $(U_1 \ldots U_N)$ followed by the inverse swaps in reverse order $(U_N^{-1} \ldots U_1^{-1})$ has trivial effective edge behavior.
\end{prop}
\begin{proof}
This follows directly by Proposition~\ref{prop:nconj} if we take $L$ to be $\mathbb{I}$.
\end{proof}

Note that $T$ above is a general `trivial' drive as defined in Sec.~\ref{sec:2D}. We can therefore continuously append or remove trivial drives from a sequence of loop drives without affecting the effective edge behavior. 

\begin{prop}\label{prop:cycle}
Given a unitary loop $L$ which is the product of a sequence of local unitary swaps $L = U_1 \ldots U_N$, then any cyclic permutation of the steps of $L$ is a loop with the same edge behavior.
\end{prop}
\begin{proof}
Consider a cyclic permutation of $L$, $L' = U_n U_{n+1} \ldots U_N U_1 \ldots U_{n-1}$. Construct the unitary $V = (U_{n} U_{n+1} \ldots U_N)^{-1}$. Then $V^{-1} L V$ is the cyclic permutation we are considering and by Proposition~\ref{prop:nconj} has the same edge behavior as $L$. 
\end{proof}

\section{Nonprimitive triangular drives}\label{app:proofarea}
In this appendix, we show that the number of independent sublattices on which a triangular drive is primitive is equal to the magnitude of its signed area (in units of the primitive triangle area). Consider an arbitrary four-step triangular drive defined by vectors $\{\bb_1,\bb_2,\bb_3,\bb_4\}$, which we take without loss of generality to have $\mathbf{b}_3 = 0$. If the triangle is not primitive, there are additional Bravais lattice points on the edges or contained within the interior of the triangle, the number of which we denote by $e$ and $i$ respectively. By specifying an edge of the triangle, we may form a parallelogram over this edge as illustrated in Fig.~\ref{fig:appB}.

\begin{figure}[t]
\centering
\includegraphics[scale=0.5]{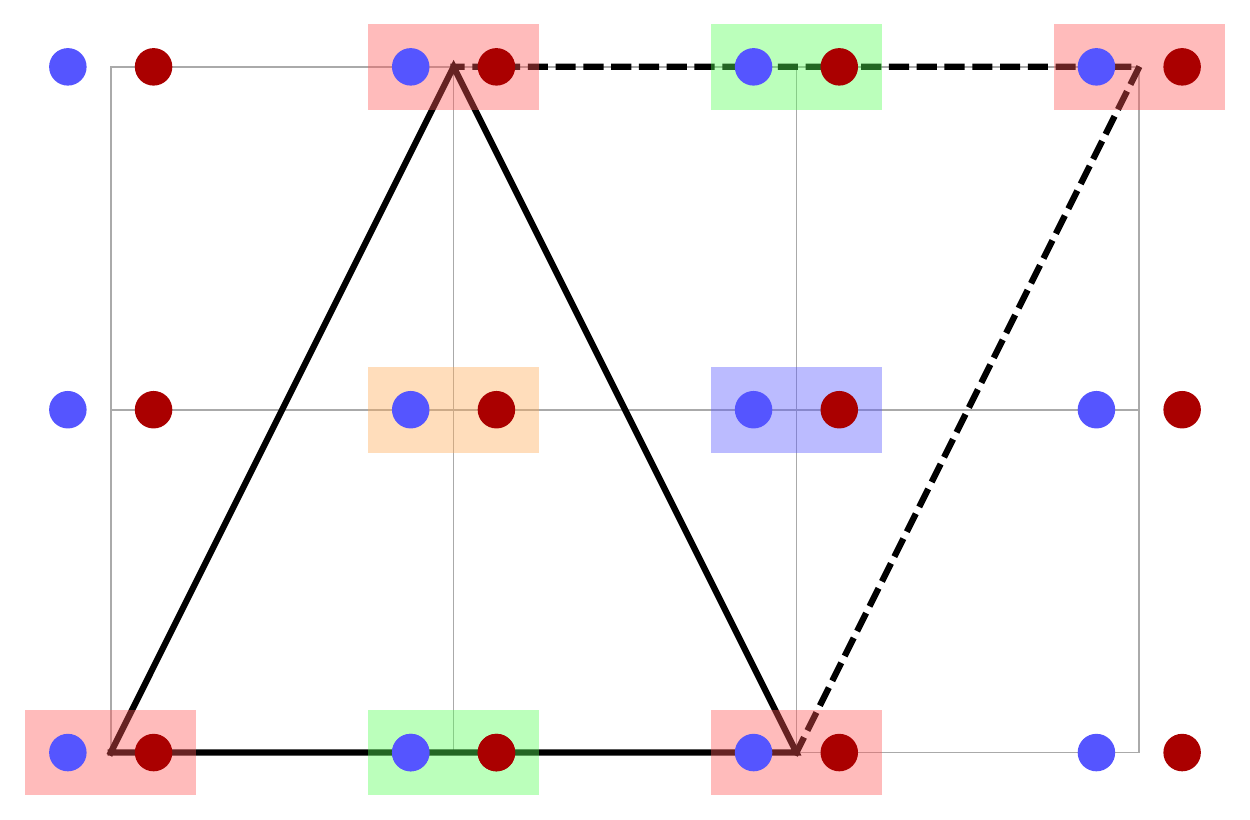}
\caption{A nonprimitive triangular exchange drive is represented by the solid lines, and a parallelogram is formed over a choice of edge. The shading of lattice sites indicates membership of different sublattices spanned by vectors of the parallelogram. The area of the solid triangle is four times the area of a primitive triangle, and there are correspondingly four different sublattices spanned by its component vectors.}
\label{fig:appB}
\end{figure}

This parallelogram may be tessellated to tile a sublattice partitioned by the drive. Each interior point of the original triangle results in two interior points of the parallelogram. Each edge point of the original triangle which lies on the edge used to construct the parallelogram results in an interior point of the parallelogram. Edge points on the other edges of the original triangle each result in two edge points of the parallelogram; however, these points are separated by a sublattice vector. By tiling the lattice with the same parallelogram but shifting the origin to these edge points and interior points, the total number of distinct sublattices spanned by the drive is found to be $1+ e + 2i$. 

Pick's theorem states that the area of a lattice polygon, in terms of the unit cell area, is given by
\be
A = v/2 + e/2 + i - 1,
\ee
where $v$ is the number of vertices. Recalling that the signed area defined in Eq.~\eqref{eq:signedarea} is given in terms of the primitive triangle area, we obtain 
\be
|A_s| &=& 2A\\
&=& 1 + e + 2i
\ee
for a triangular drive. Hence, the number of independent sublattices is equal to the magnitude of the signed area of the drive. Since each independent sublattice generates its own edge behavior, the edge behavior of a triangular drive is equivalent to a composition of $|A_s|$ primitive drives.

\end{document}